\documentclass[11pt]{article}
\usepackage{amssymb,amsmath,amsthm}
\usepackage{bbm}
\usepackage[noend]{algpseudocode}
\usepackage{graphicx}
\usepackage{verbatim}
\usepackage{url,xspace,hyperref}
\usepackage{caption}
\usepackage{subcaption}
\usepackage{multirow}
\usepackage{multicol}
\usepackage{latexsym}
\usepackage{amsmath,amssymb,enumerate}
\usepackage{xspace}
\usepackage{algorithm,algorithmicx}
\usepackage{float}
\usepackage{xcolor}
\usepackage{mathrsfs}
\usepackage{cleveref}
\usepackage{float}
\usepackage{mathtools}
\usepackage{bm}
\usepackage{tikz}
\usepackage{caption}
\usepackage{multirow}
\usepackage[normalem]{ulem}
\usepackage{makecell}
\usepackage{enumitem}  
\usepackage[round]{natbib}
\usetikzlibrary{positioning,decorations.pathreplacing,quotes}
\usepackage{setspace}

\tikzset{tick/.style={draw, minimum width=0pt, minimum height=2pt, inner sep=0pt, label=below:$#1$},
    tick/.default={}}

\usepackage{fullpage}
\usepackage{geometry}
\usepackage{setspace}
\usepackage{changepage}
\usepackage{titlesec}

\hypersetup{
  colorlinks   = true, %Colours links instead of ugly boxes
  urlcolor     = blue, %Colour for external hyperlinks
  linkcolor    = black, %Colour of internal links
  citecolor   = red %Colour of citations
}

\newtheorem{theorem}{Theorem}[section]
\newtheorem*{theorem*}{Theorem}

\newtheorem{proposition}[theorem]{Proposition}
\newtheorem{claim}[theorem]{Claim}
\newtheorem{lemma}[theorem]{Lemma}

\newtheorem{example}{Example}[section]

\newtheorem{remark}[theorem]{Remark}

\newcommand{\ppp}{\textsc{OS-UD}\xspace}

\newcommand{\E}{\mathbb{E}} 
\newcommand{\OPT}{\mathtt{OPT}}
\newcommand{\ALG}{\mathtt{ALG}}

\allowdisplaybreaks

\title{\vspace{-0.5in} Online Selection with Uncertain Disruption}
 \author{
 Yihua Xu\thanks{Department of Computational Applied Mathematics and Operations Research, Rice University, USA.}\ 
 \thanks{Ken Kennedy Institute, Rice University, USA.} \and
 Süleyman Kerimov\thanks{Jones Graduate School of Business, Rice University, USA.}\ \footnotemark[2] \and
 Sebastian Perez-Salazar\footnotemark[1]\ \footnotemark[2]
 }

\date{\today}

\begin{document}

\maketitle

\vspace{-0.2in}

\begin{abstract}
In numerous online selection problems, decision-makers (DMs) must allocate on the fly limited resources to customers with uncertain values. The DM faces the tension between allocating resources to currently observed values and saving them for potentially better, unobserved values in the future. Addressing this tension becomes more demanding if an uncertain disruption occurs while serving customers. Without any disruption, the DM gets access to the capacity information to serve customers throughout the time horizon. However, with uncertain disruption, the DM must act more cautiously due to risk of running out of capacity abruptly or misusing the resources. Motivated by this tension, we introduce the Online Selection with Uncertain Disruption (OS-UD) problem. In OS-UD, a DM sequentially observes $n$ non-negative values drawn from a common distribution and must commit to select or reject each value in real time, without revisiting past values. The disruption is modeled as a Bernoulli random variable with probability $p$ each time DM selects a value. We aim to design an online algorithm that maximizes the expected sum of selected values before a disruption occurs, if any.

We evaluate online algorithms using the \emph{competitive ratio}---the ratio between the expected value achieved by the algorithm and that of an optimal \emph{clairvoyant} algorithm that knows all value realizations in advance but still faces uncertain disruption.  Using a quantile-based approach, we devise a \emph{non-adaptive} single-threshold algorithm that attains a competitive ratio of at least $1-1/e$, and an \emph{adaptive} threshold algorithm characterized by a sequence of non-increasing thresholds that attains an asymptotic competitive ratio of at least $0.745$. Both of these results are worst-case optimal within their corresponding class of algorithms and continue to hold regardless of whether the last value is partially recoverable. Our results reveal an interesting connection between the OS-UD problem and the i.i.d.\ prophet inequality problems as the number of customers grows large. 

\end{abstract}

\section{Introduction} \label{sec:Intro}

Online selection models have gained increasing attention, as they capture key features of extensive applications such as online advertising (\cite{mehta2012online}), online resource allocation (\cite{delong2022online}), and applicant evaluation (\cite{epstein2024selection}). Generally speaking, in these problems, a decision-maker (DM) must irrevocably allocate limited resources to incoming customers with uncertain values. The DM faces a fundamental tension between allocating resources to currently observed values and saving them for potentially better, unobserved values. 

% Preamble (if not already present)
% \usepackage[normalem]{ulem}  % for \sout
% \usepackage{xcolor}          % for \textcolor
% \newcommand{\ppp}{\textsc{OSUD}} % example; define as you like

%\sout{In a similar vein, an owner of a reusable resource (e.g., a host renting her property on Airbnb, or a driver accepting a ride on Uber) might experience a disruption in operations due to property damage caused by the customer. In general, this additional uncertainty forces the DM to act more cautiously due to risk of misusing the resources.}

Alleviating this tension becomes more challenging if committing to a request may cause a \emph{disruption}, potentially halting the remaining process. For example, in cloud computing, providers such as Amazon Web Services and Azure often offer spot instances at lower prices than on-demand instances by utilizing underutilized resources. While these providers aim to maximize resource utilization and generate additional revenue through low-priced spot instances,  providers must manage the risk of reclaiming these resources for higher-priority on-demand users; when reclamation occurs, service is interrupted and the provider typically recovers only a fraction of the value from the disrupted task (e.g., see \cite{cohen2019overcommitment,perez2022dynamic}). In a similar vein, an owner of a reusable resource (e.g., a host renting her property on Airbnb, or a driver accepting a ride on Uber) might experience a disruption in operations due to property damage or other unforeseen issues caused by the customer. While in practice such disruptions are often temporary—for instance, an Airbnb host may suspend operations during repairs and then resume hosting—one can view these events as a sequence of disruptions and re-entries. Our model captures the cost of these interruptions through the abstraction of disruptions that halt service, regardless of whether they are permanent or repeated. In this context, the last accepted transaction typically yields its full value before any subsequent downtime. Motivated by this, we introduce the \emph{Online Selection with Uncertain Disruption} (\ppp) problem, which captures such disruptions when serving incoming requests.% \textcolor{blue}{and the value from the disrupted task may be partially recoverable.}

% Notably, the last accepted transaction typically yields its full value (the guest paid, the ride completed) before any subsequent downtime.

In the \ppp problem, a finite sequence of $n$ independent and identically distributed (i.i.d.) non-negative random variables is observed sequentially. Upon observing a value, the DM must decide whether to (i) select it, or, (ii) reject it and observe the next value, if any, without the possibility of reconsidering past observations. 
% \sout{However, with a known probability $p\in [0,1]$, selecting a value disrupts the process, forfeiting this last selection, and halting the remaining process permanently.} 
{Importantly, with a known probability $p \in [0,1]$, selecting a value triggers a disruption that halts the process, and the DM partially recovers the last selected value; otherwise, with probability $1-p$, the selection is retained and the process continues with the remaining observations.}

%{  Importantly, with a known probability $p \in [0,1]$, selecting a value triggers a disruption that halts the process, but the last selected item \textit{possibly} contributes a fixed fraction of its value; otherwise, with probability $1-p$, the selection is retained and the process continues with the remaining observations.}

%Our objective is to design an online algorithm that maximizes the expected sum of selected values before a disruption occurs, if any.

{Our objective is to design an online algorithm that maximizes the expected sum of selected values under disruption. We evaluate the performance of algorithms using the \emph{competitive ratio}, which is the fraction between the value of an algorithm and the value of an optimal \emph{clairvoyant algorithm} that \emph{knows the values upfront and can select order of observation}, but still faces \emph{unknown disruptions}. 
The competitive ratio is a number in the interval $[0,1]$ that measures the ``price" paid by the DM for not knowing all the sequence of values upfront. Our definition of competitive ratio differs slightly from the usual definition found in other online selection problems (e.g., see \cite{krengel1977semiamarts,mehta2012online,hill1982comparisons,correa2021posted})---a formal description of our problem and discussion is presented in \S\ref{subsec:formulation}. 
Our goal is to shed light on the following questions: (i) What is the optimal (worst-case) competitive ratio for the \ppp problem? (ii) How should a DM accept incoming requests to achieve the optimal competitive ratio?

In general, without imposing any structure on the disruption process, any online algorithm can perform arbitrarily poorly and the competitive ratio can be $0$ (see Example \ref{example:vary_p}). In this paper, we answer both questions by devising simple threshold policies, and we provide an exhaustive analysis for the case when the probability of disruption $p$ is a fixed constant. We first show that the optimal online algorithm induced by a stochastic dynamic program is characterized by a sequence of non-increasing thresholds, where the algorithm accepts the incoming request if and only if its value is at least the current threshold. Motivated by this, in order to quantify the ``price" that the DM must pay due to uncertainties, we analyze two classes of algorithms that are easy to describe: \emph{non-adaptive (fixed) threshold algorithms} and \emph{adaptive threshold algorithms}. Such algorithms are also desirable in practice, as they are simple-to-implement, have economic interpretations~\citep{arnosti2023tight,van1995dynamic,naor1969regulation}, and are widely used in posted price mechanisms~\citep{chawla2007algorithmic,chawla2010multi,chawla2024static,correa2019pricing}.

For non-adaptive threshold algorithms, we offer a complete characterization of the optimal competitive ratio, which is $1 - 1/e$, through the design of an algorithm and a hard instance. For adaptive threshold algorithms, we present an algorithm that employs a non-increasing sequence of thresholds with a tight asymptotic competitive ratio of $\theta^* \approx 0.745$, where $\theta^*$ is a parameter appearing in the Hill and Kertz equation \citep{hill1982comparisons,kertz1986stop}. Interestingly, $\theta^*$ coincides with the optimal competitive ratio of the i.i.d.\ prophet inequality problem \citep{correa2021posted,hill1982comparisons,kertz1986stop}.

The techniques derived for the fixed disruption probability case can also be adapted to the case when disruption is rare (where $p=\alpha/n$, for $\alpha < 1$). Notably, one can obtain asymptotic competitive ratios that are strictly larger than $\theta^*$ (see \S\ref{sec:final_rem}). After presenting the formal introduction of our problem in the next subsection, we provide the details of our results and techniques in \S\ref{subsec:contri}.

\subsection{Problem Formulation}\label{subsec:formulation}

Given a fixed disruption probability $p\in [0,1]$ and a partial recovery parameter $\zeta \in [0,1]$, an instance $I$ of the \ppp\ problem is given by a pair $(n,F)$, where $n \geq 1$ is the number of values, and $F$ is a continuous and increasing cumulative distribution function supported in the non-negative real numbers.\footnote{We can assume that $F$ is continuous and smooth by perturbing the observed values with a small continuous noise. The monotonicity assumption is less common, but follows a similar principle of perturbing the original $F$ (e.g., see~\cite{liu2020variable,perezsalazar2024optimalguaranteesonlineselection}).} Let $\mathcal{F}$ be the set of all such cumulative distribution functions. An instance $I=(n,F)$ encodes two sequences of random variables $(X_i)_{i=1}^n$ and $(Y_i)_{i=1}^n$. The $X_i$'s are non-negative and i.i.d.\ following the distribution $F$. The $Y_i$'s are are i.i.d.\ $0/1$-random variables with $\Pr[Y_i = 1] = p \in [0,1]$. The DM observes $X_i$'s in a fixed order $1,\ldots,n,$ sequentially. Upon observing \(X_i\), the DM can either skip to the next value (if one remains) or select the current one. If the DM selects \(X_i\), the DM then observes \(Y_{i}\) indicating whether making this selection will disrupt the whole process. If \(Y_{i} = 1\), we say that a \emph{disruption occurred} and the process terminates. In this case, the DM receives the sum of selected values up to the disruption and a fraction $\zeta$ of the last selected value $X_i$. If \(Y_{i} = 0\), the process continues, and the DM receives \(X_i\) as part of the current round's value. We aim to find an online algorithm $\ALG$ that maximizes the total expected value collected. We denote by $v(\ALG(I))$ the expected value obtained by algorithm $\ALG$ for an instance $I$.

We adopt a competitive analysis, where we compare the value of an algorithm to an \emph{optimal clairvoyant algorithm} that has access to the value realizations $X_1,\ldots,X_n$, in advance, \emph{but not to the sequence $Y_{1},\ldots,Y_{n}$}. 
% \sout{The clairvoyant can choose in which order to observe the $X_i$'s, say $X_{\sigma(1)},\ldots,X_{\sigma(n)}$, where $\sigma$  is a permutation of $[n]$. }
The clairvoyant does not operate sequentially in time; rather, it can freely choose the order in which to examine the values, denoted $X_{\sigma(1)},\ldots,X_{\sigma(n)}$, where $\sigma$ is a permutation of $[n]$. In this sense, the clairvoyant \emph{searches} over the available options without being constrained by the timeline of arrivals, thereby representing the best possible performance under full knowledge of values, but uncertainty about disruptions.
Thus, the optimal clairvoyant algorithm $\OPT$ obtains a value of {the top $D-1$ order statistics before disruption, and additionally receives $\zeta$ of the last item if disruption occurs within $n$ values:}

\begin{align}
v(\OPT (I))&=\E\left[ 
\max_{\sigma} \left\{ \sum_{i=1}^{\min\{ D-1,n\}} X_{\sigma(i)} +
\zeta \cdot \mathbb{I}\{D\leq n \} X_{\sigma(D)} \right\}\right] \notag \\
&=\E\left[ \sum_{i=1}^{\min\{D-1,n \} } X_{(i)} + 
\zeta \cdot \mathbb{I}\{D\leq n\}\,X_{(D)}\right] \label{exp:opt_prilim},
\end{align}
where $D$ is a geometric random variable with parameter $p$,\footnote{That is, $\Pr[D=j]=p(1-p)^{j-1}$ for $j \geq 1$.} and $X_{(j)}$ denotes the $j^{\text{th}}$ largest order statistic. 
% {  We note that a benchmark that knows in advance how many selections it can make before the disruption occurs will also take the same actions as OPT.}
We abuse notation and simply write $\ALG$ and $\OPT$ for $\ALG(I)$ and $\OPT(I)$, respectively, when the instance $I$ is clear from the context. For convenience, we will also omit the word ``online'' when referring to online algorithms when the context is clear. Then, given an algorithm $\ALG$ with a disruption probability $p$ and {a partial recovery parameter $\zeta$}, its \emph{competitive ratio} is $$\inf_{n \geq 1, F \in \mathcal{F}}\frac{v(\ALG)}{v(\OPT)},$$ and we seek an algorithm with the largest competitive ratio possible. {In the remainder of the paper, we focus on the case $p \in (0,1)$, since: (i) when $p=0$, there are no failures and it is optimal to accept all values; (ii) when $p=1$ and $\zeta=0$, both the clairvoyant algorithm and any other algorithm obtain a value of $0$; and (iii) when $p=1$ and $\zeta \in (0,1]$, the problem reduces to the classic i.i.d.\ prophet inequality problem with values rescaled by $1/\zeta$.
%or $p = 1$ with $\zeta = 0$, any optimal algorithm obtains the same value as the optimal clairvoyant algorithm, and the model reduces to the classical single-choice prophet problem when $p=1$ with $\zeta \in (0,1]$. In the latter case, both the optimal clairvoyant and online algorithms incur a fraction $(1-\zeta)$ loss, leaving the competitive ratio unchanged from the classical prophet setting.
}

We remark that the clairvoyant algorithm is different from the \emph{offline algorithm} that knows the realizations of $X_i$'s and $Y_{i}$'s upfront for all $i\leq n$. The following example shows that it is impossible to obtain a positive constant competitive ratio when we attempt to replace $\OPT$ by the offline algorithm.

\begin{example}\label{benchmarkexample}
    Consider the following $n \geq 1$ i.i.d.\ values $X_1,\ldots,X_n$, where $X_i$'s are uniformly distributed in $[0,n]$. Fix $p \in (0,1)$ {and a partial recovery parameter $\zeta = 0$}. Note that the offline algorithm will select all the values for which $Y_{i}=0$. Thus, the expected value obtained by the offline algorithm is $v(\OPT^{\mathrm{offline}}):=\E[\text{number of $i$'s with $Y_{i}=0$}] \cdot \E[X_1] = (1-p) n^2/2$. 
    On the other hand, 
    we have $v(\OPT) \leq n/p$. 
    Then, $v(\ALG)/v(\OPT^{\mathrm{offline}})\leq v(\OPT)/v(\OPT^{\mathrm{offline}})\leq 2/(p(1-p)n)$, for any algorithm $\ALG$. This shows that, for any fixed $p\in (0,1)$, no constant competitive ratio is possible when we compare $v(\ALG)$ against the value of the offline algorithm that knows all the information in advance.
\end{example}

Note that one can formulate the DM's problem as a stochastic dynamic program (see \S \ref{sec: prelim}). However, analyzing competitive ratios is not straightforward as we are comparing two algorithms that operate with asymmetric information. Nevertheless, the dynamic program reveals a nice structural property that we will leverage in our analysis. Indeed, given an instance $I=(n,F)$ of the \ppp problem, there exists an optimal algorithm that employs thresholds $\tau_1 \geq \cdots \geq \tau_n$ such that if the $i^{\text{th}}$ observed value $X_i$ is at least $\tau_i$, then the algorithm selects it (see Proposition~\ref{prop:decreasing_thres}). This motivates us to focus on the two
extremes for this class of threshold-based algorithms: (1) NA, which is the class of \emph{non-adaptive algorithms}, where $\tau_1=\cdots=\tau_n$; and (2) AD, which is the class of \emph{adaptive algorithms}, where we do not constrain the thresholds to be the same, but the thresholds must be non-increasing, i.e., $\tau_1 \geq \cdots \geq \tau_n$. For a fixed $p\in (0,1)$ {and  $\zeta \in [0,1]$}, and a class of algorithms $\mathcal{C}\in \{ \mathrm{NA}, \mathrm{AD} \}$, we define the competitive ratio for instances of length $n$ in class $\mathcal{C}$ via

\[
\gamma_{n,p,{ \zeta}}^{\mathcal{C}} := \hspace{0.12cm} \sup_{\ALG \text{ in } \mathcal{C}} \hspace{0.12cm} \inf_{F \in \mathcal{F}} \hspace{0.12cm}  \frac{v(\ALG(n,F))}{v(\OPT(n,F))}. 
\]
Given that $\mathcal{C}$ is a class of threshold-based algorithms, we can exchange the supremum with the infimum in the definition of $\gamma_{n,p, \zeta}^{\mathcal{C}}$ without altering its value. Thus, the largest competitive ratio in the class $\mathcal{C}$ for the \ppp problem can be written as follows:
\[
\gamma_{p, { \zeta}}^{\mathcal{C}} :=  \hspace{0.12cm} \sup_{\ALG \text{ in } \mathcal{C}} \hspace{0.12cm} \inf_{n \geq 1, F \in \mathcal{F}} \hspace{0.12cm}  \frac{v(\ALG)}{v(\OPT)} = \inf_{n\geq 1} \gamma_{n,p,{  \zeta}}^{\mathcal{C}}.
\]

\subsection{Our Technical Contributions} \label{subsec:contri}

Our first result provides a tight performance bound for the non-adaptive threshold algorithms. 

\begin{theorem}\label{thm:static}
For any $p\in (0,1)$ {and  $\zeta \in [0,1]$}, we have $\gamma_{p, \zeta}^{\mathrm{NA}}= 1-1/e$.
\end{theorem}
To prove Theorem~\ref{thm:static}, we explicitly construct a fixed threshold algorithm with a competitive ratio of at least $1-1/e$, through focusing on a subclass of NA, the \textit{quantile-based} threshold algorithms. Algorithms in this subclass are parametrized by $q\in [0,1]$, which we refer as the quantile of the algorithm. An algorithm with a quantile $q$ computes a threshold $\tau$ via $q=\Pr[X\geq \tau]=1-F(\tau)$, and selects any value in the input sequence that is at least $\tau$. 

To prove that the competitive ratio of an algorithm $\ALG$ is at least $\theta$, one can use standard stochastic dominance arguments, and it is sufficient to show that $\Pr \left(\ALG \text{ obtains value of at least }\tau \right)$ is at least  $\theta \cdot \Pr \left(\OPT \text{ obtains value of at least }\tau\right)$ for any $\tau\geq 0$. However, due to possibility of making multiple selections, analyzing both of these quantities is not evident. Instead, we characterize the optimal values for both the clairvoyant and the quantile-based threshold algorithms as integrals of the inverse of $F$. {From here, we can deduce a lower bound for $\gamma_{n,p,\zeta}^{\mathrm{NA}}$ that only depends on $n$ and $p$, and not dependent on $\zeta$ or $F$.} This lower bound turns out to be monotonically decreasing in $n$ and converging to $1-1/e$, providing the desired lower bound on $\gamma_{p, \zeta}^{\mathrm{NA}}$. {To prove an upper bound on $\gamma_{p, \zeta}^{\mathrm{NA}}$, we explicitly construct an instance $(n,F)$ such that $\gamma_{n,p,0}^{\mathrm{NA}}\leq 1-1/e+o(1)$, and since $\gamma_{p, \zeta}^{\mathrm{NA}}=\inf_{n\geq 1} \gamma_{n,p,\zeta}^{\mathrm{NA}}$, the result follows.} 
{For ease of exposition, we present the proof for the case $\zeta = 0$ in \S\ref{sec:static}, and we treat the general case $\zeta \in [0,1]$ in \S\ref{sec:recovery_na}. }

%{  For ease of exposition, we present the proof for the case $\zeta = 0$ in \S\ref{sec:static}, and we treat the general case $\zeta \in [0,1]$ in \S\ref{sec:recovery_na}. }

%We present the proof for the case $\zeta = 0$ in \S\ref{sec:static}, which corresponds to no partial recovery on the last accepted item and thus a clean baseline model that best illustrates the methodology. We then treat the general case $\zeta \in [0,1]$ in \S\ref{sec:recovery_na}. }

We then turn our attention to the performance of the general class of adaptive threshold algorithms, and we show that \textit{adaptivity} is key. Our second result shows that a competitive ratio better than $1-1/e$ is possible for such algorithms.

\begin{theorem}\label{thm:adaptive}
For any $p\in (0,1)$ and {  $\zeta \in [0,1]$}, $\liminf_{n \to \infty} \gamma_{n, p, \zeta}^{\mathrm{AD}}\geq \theta^*\approx 0.745$, where $\beta=1/\theta^*$ is the unique solution to the integral equation $\int_0^1 (y-y\ln y+(\beta-1))^{-1}\, \mathrm{d}y=1$. 
\end{theorem}

Theorem~\ref{thm:adaptive} shows that there is a strict separation between acting adaptive and non-adaptive. In order to prove this result, we provide a quantile-based algorithm similar to~\cite{correa2021posted} and \cite{perez2025iid}. For each observed value $i$, we sample a quantile $q_i$ from an appropriate distribution, and we compute a threshold $\tau_i$ such that $\Pr[X\geq \tau_i]=1-F(\tau_i)=q_i$. The algorithm then selects $i$ if the observed value $X_i$ is at least $\tau_i$. By an appropriate choice of distribution for the quantiles, we show that the value of the algorithm is asymptotically a fraction $\theta^*\approx 0.745$ of $v(\OPT)$. { Again for ease of exposition, we present the algorithm and its analysis for $\zeta = 0$ in \S\ref{sec:adapt}, and analyze the general case $\zeta \in [0,1]$ in \S\ref{sec:recovery_a}. }

Our final result establishes a limit on the competitive ratio attainable by any algorithm for the $\ppp$ problem, contributing another piece to the puzzle of answering our first research question.

\begin{theorem}\label{thm:adaptive_ub}
    For any $p\in (0,1)$ { and  $\zeta \in [0,1]$}, we have $\gamma_{p, \zeta}^{\mathrm{AD}}\leq \theta^*$, where $\theta^*$ is defined in Theorem~\ref{thm:adaptive}.
\end{theorem}

We derive the upper bound by adapting the worst-case instance for the i.i.d.\ prophet inequality problem from~\citep{liu2020variable,hill1982comparisons} to \ppp. We present the details in \S\ref{sec:adapt}. {  Taken together, Theorems~\ref{thm:adaptive} and~\ref{thm:adaptive_ub} imply that, for any fixed \(p\in(0,1)\),
\(\lim_{n\to\infty}\gamma_{n,p, \zeta}^{\mathrm{AD}}=\theta^*\approx 0.745\).}

{  Our quantile-based techniques are highly flexible. In \S\ref{sec:extensions}, we show that the single-threshold approach extends to a variant of the \ppp{} problem, in which the objective is not to maximize the sum of the selected values, but rather to \emph{maximize the largest value} among them. For this model, we prove a tight competitive ratio of $1 - e^{-\lambda(p)}$ for the class of non-adaptive algorithms, where $\lambda=\lambda(p)$ is the unique solution to the equation $(1 - e^{-\lambda p}) = p\lambda (1 - e^{-\lambda})$. Notably, this competitive ratio lies in the range $[1 - 1/e, 1]$ and approaches $1$ as $p \to 0$. A formal definition of this model, along with our results and analysis, is provided in \S\ref{sec:extensions}. }

% \yx{I think maximize the maximum is a bit weird, how about maximize the largest? Otherwise looks good}

\section{Related Work}

Online selection, as well as the tension between collecting short-term values and saving resources for long-term values, have been extensively studied in computer science and operations research literature through the lens of optimal stopping theory (e.g., see~\cite{hill1992survey, shiryaev2007optimal, krengel1977semiamarts}), and for their broad applicability on various practical problems from crowd-sourcing \citep{mehta2012online} to capital investment problem \citep{goyal2010ptas}. Here, we present several streams of literature that are closely related to our work.

\noindent{\textbf{Prophet Inequality.}} 
One related stream of literature to our work is the prophet inequality problem (e.g., see \cite{krengel1977semiamarts,samuel1984comparison}), in particular, the case with i.i.d.\ values. In the i.i.d.\ prophet inequality problem \citep{hill1982comparisons}, a DM must select at most one value from a sequence of i.i.d.\ randomly generated values, and her goal is to design an algorithm with a large competitive ratio, where the offline benchmark is the expected maximum of the sequence of values. It is known that the optimal competitive ratio for the i.i.d.\ prophet inequality problem is $\theta^* \approx 0.745$, which is the unique parameter appearing in the Hill \& Kertz equation~\citep{hill1982comparisons,kertz1986stop,correa2021posted}. This optimal algorithm can be attained by a quantile-based algorithm that depends solely on $n$ and is independent of the specific instance distribution. To facilitate the analysis of the algorithm, the competitive ratio can be retrieved through a unique solution of the Hill and Kertz equation. Similar techniques to using quantile algorithms (e.g., see \cite{feng2025iid, allouah2023optimal}) and ODEs have also been explored in various recent work (e.g., see \cite{liu2020variable, correa2021posted}). By generalizing the Hill and Kertz equation, \cite{brustle2024splitting} introduces a novel non-linear system of differential equations and provide tight analysis for the $k$-prophet inequality problem. Our paper extends the single-selection prophet setting by incorporating an additional disruption indicator, and we use a similar quantile-based approach to develop an adaptive threshold algorithm, which achieves an asymptotic competitive ratio of $\theta^*$ as well.

\noindent {\bfseries Random horizon.} There is large body of work in optimal stopping problems with random horizon (e.g., see \citep{hajiaghayi2007automated,zhang2023secretary}). Here, the disruption is caused by not knowing the length $n$ of values upfront. For example, the uncertain horizon setting has been extensively studied within the framework of the secretary problem. When no distributional information is available regarding the disruption time---beyond which no further applicants can be picked---it is known that no algorithm can achieve a constant competitive ratio~\citep{hill1991minimax}. However, if a random termination time with a known value-independent distribution exists, a conditionally optimal selection rule can be formulated~\citep{samuel1996optimal}. In our case, the disruption is caused potentially by selecting a value, and in principle we could observe the whole sequence of $n$ values. Closer to our work is~\cite{alijani2020predict}, which studies the prophet inequality problem with supply uncertainty. Even though the \ppp problem has applications in settings with supply uncertainty, our main focus is in applications, where serving a request can disrupt the remaining selection process. In a similar vein, our model is loosely related to  \emph{stochastic knapsack} problems (e.g., see \citep{dean2008approximating,ma2018improvements}), where items have unknown stochastic sizes, items are packed sequentially, and if the knapsack is overflowed, then the whole process stops.
We could regard the \ppp problem as a knapsack problem, where we have a knapsack of capacity $1$, and each item has two possible sizes: size $1/n$ with probability $1-p$ and size $1+1/n$ with probability $p$. Nevertheless, the knapsack literature mostly deals with approximation algorithms as opposed to competitive analysis that we pursue in this work.

\noindent{\bfseries Dynamic matching and online resource allocation.} The tension between committing to a decision now and delaying decisions in anticipation of better opportunities arises as an inherent trade-off in many stochastic models. In the context of dynamic matching and online resource allocation, recent work addressed this trade-off through the lens of an all-time regret notion \citep{wei2023constant, kerimov2024dynamic, kerimov2025optimality, he2025online, gupta2024greedy} and approximation algorithms \citep{aouad2020dynamic}. In particular, the all-time regret notion implicitly deals with uncertain disruption by guaranteeing near-optimal performance throughout the time horizon. Our work studies the same trade-off in the context of online selection by explicitly introducing a disruption indicator, and while there are differences across these stochastic models, we hope that the modeling we propose in this work can be leveraged to study this fundamental tension under disruptions such as match rejections and item returns. 

\section{Preliminaries} \label{sec: prelim}

In this section, we present some preliminaries needed in the remainder of the paper. All missing proofs in the main body are deferred to the appendix. The value maximization problem faced by the DM can be solved by means of stochastic dynamic programming. For any $n \geq 1$, a disruption parameter $p \in (0,1)$, {  a partial recovery parameter $\zeta \in [0,1]$}, and a distribution $F \in \mathcal{F}$, let $D_{i}(p,F)$ be the optimal value obtainable from the sequence $(X_j)_{j=i}^n$, when $i-1 \leq n$ values have been observed already (i.e., the $i^\text{th}$ value is ready to be observed next) and no disruption has occurred yet. Then clearly $D_{n+1} (p,F) = 0$, and for $i\leq n$, we have
% \begin{equation} \label{exp:dynamic_programming}
%     D_{i}(p,F) = \max_{\tau_i \geq 0} \Bigg\{(1-p) \Pr[X\geq \tau_i] \Big(\E[X \mid X\geq \tau_i] + D_{i+1}(p, F)\Big)+ \Pr[X < \tau_i] D_{i+1}(p, F)\Bigg\}.
%     \tag{DP}
% \end{equation}
{ 
\begin{equation*}\label{exp:dynamic_programming}
\begin{aligned}
D_i(p,F) \; = \max_{\tau_i \geq 0} & \Big\{ 
\Pr[X < \tau_i]\, D_{i+1}(p,F) \hspace{9cm} \text{(DP)} \\
& + \Pr[X \geq \tau_i] \Big( (1-p)\big( \mathbb{E}[X \mid X \geq \tau_i] + D_{i+1}(p,F)\big)
+ p\,\zeta\, \mathbb{E}[X \mid X \geq \tau_i] \Big)
\Big\}.
\end{aligned}
\end{equation*}}
\noindent By solving (DP), one can obtain an optimal algorithm that at time $i \in [n]$, selects the observed value 
$X_i$ if it is at least $\tau_i$, where $\tau_i$ is the maximizer of (DP). The next proposition states that it is optimal to use a non-increasing sequence of thresholds, i.e., $\tau_1 \geq \cdots \geq \tau_n$.

\begin{proposition}\label{prop:decreasing_thres}
There exists an optimal algorithm for the \ppp problem that employs a non-increasing sequence of thresholds. 
\end{proposition}

%\noindent { The proof is based on standard swapping thresholds argument and given in Appendix \ref{app:sec0}.} 

Next, we provide two characterizations of $v(\OPT)$ in terms of the inverse of $F$ and its derivative. These characterizations will be useful in the analysis of NA and AD classes. 
% {  Both representations of $v(\OPT)$ originate from the order-statistic perspective. Starting with $\E[X_{(j)}]$, one may either (i) rewrite it as a binomial-tail integral, or (ii) expand the order-statistic formula in quantiles, swap the order of summation and integration, and apply the binomial theorem.}
\begin{proposition} \label{prop:opt_value_complete}
The value achieved by the optimal algorithm can be characterized as follows:
    \begin{equation} \label{exp:opt_value_complete}
    v(\OPT) = \int_{0}^1 { \left(B_n (p, v) + \zeta \big[1 - (1 - vp)^n] \right)} r(v) \, \mathrm{d}v = \int_{0}^1 F^{-1}(1-q) { g_n^{(\zeta)}(p, q)} \, \mathrm{d}q, 
\end{equation}
where $B_n (p, v) := \E [\min \{X, D-1, n\} ]$ with  $X \sim \text{Bin}(n,v)$, and $r(v) > 0$ is a function such that $\int_{u}^1 r(v) \, \mathrm{d}v = F^{-1}(1-u)$,\footnote{The existence of $r(v)$ is guaranteed by our assumption of $F^{-1}$ being differentiable and strictly decreasing.} and { $g_n^{(\zeta)}(p,q) \;=\; n\bigl(1-(1-\zeta)p\bigr)\,(1-pq)^{n-1}$. }
\end{proposition}

{  Both representations of 
$v(\OPT)$ originate from its characterization in \eqref{exp:opt_prilim}, which express the value as the expected sum of order statistics up to the disruption point:
\begin{align*}
v(\OPT(I))
&= \Pr[D>n]\;\mathbb{E}\!\left[\sum_{i=1}^{n} X_{(i)}\right]
  \;+\; \sum_{j=1}^{n} \Pr[D=j]\;
      \mathbb{E}\!\left[\sum_{i=1}^{j-1} X_{(i)} \;+\; \zeta\,X_{(j)}\right].
\end{align*}
For each term $\E[X_{(i)}]$, two equivalent characterizations are available: (i) the tail-integral representation, where the integrand is the survival function, or, (ii) the density-based representation, where one starts from the distribution of order statistics, performs a change of variable into quantile form, interchanges integrals, and applies the binomial theorem. The formal proof is given in \S\ref{app:sec0}.}

Finally, we reproduce the system of ordinary differential equation (ODE) introduced by~\cite{hill1982comparisons}, colloquially known as the Hill and Kertz equation. This will be used in the analysis of the AD class. We want to find a solution $y:[0,1]\to [0 ,1]$ to the following ODE:
\begin{align*}
    \label{hill_kertz}
y'(t) &= y(t) (\ln(y(t))-1) - 
\left(\frac{1}{\theta^*} -1 \right),\\
y(0) &=1,\quad \lim_{t\uparrow 1}y(t) = 0. 
\end{align*} 
This system has a solution if and only if $\theta^*\approx 0.745$, where $\theta^*$ the unique solution to the integral equation in Theorem~\ref{thm:adaptive}.

\section{The Class of Non-Adaptive Algorithms}\label{sec:static}

In this section, we analyze the performance of non-adaptive threshold algorithms { for $\zeta = 0$}. We fully characterize the competitive ratio achievable by this class, presenting an optimal algorithm that guarantees a competitive ratio of at least $1 - 1/e$ for any input to the \ppp problem. Moreover, we show that no non-adaptive algorithm can achieve a competitive ratio exceeding $1 - 1/e + o(1)$ in the worst case. Together, these results imply Theorem~\ref{thm:static}.

\subsection{Quantile-Based Non-Adaptive Algorithm}\label{subsec:static_cr}

Our algorithm is quantile-based. It receives a quantile $q\in [0,1]$, sets the threshold $\tau = F^{-1} (1-q)$, and selects any value of at least $\tau$. For any $n \geq 1$, we denote the algorithm with quantile $q_n$ by $\ALG^{q_n}$ and its expected value by $v(\ALG^{q_n})$. The main result of this section is the following:

\begin{theorem}\label{thm:lb_for_quantile}
    For any $n\geq 1$, if ${q}_n=\min\{1,1/pn\}$, we have $\frac{v(\ALG^{q_n})}{v(\OPT)} \geq 1-\frac{1}{e}$.
\end{theorem}

\noindent Theorem \ref{thm:lb_for_quantile} immediately implies $\gamma_{p,{ 0}}^{NA}\geq 1-1/e$, proving the first part of Theorem~\ref{thm:static}. The proof relies on two key lemmas. The first lemma provides a lower bound on $v(\ALG^{q_n})/v(\OPT)$ that is independent of $F$.

\begin{lemma}{\label{lem:general_formula}}
    For any $n \geq 1$ and $p \in (0,1)$, we have
    \begin{equation}
        \frac{v(\ALG^{q_n})}{v(\OPT)} \geq \frac{(1-(1-q_np)^n)}{p}\cdot\min\left\{  \frac{p}{1-(1-p)^n}, \frac{1}{q_n n}  \right\}, \label{eq:lower_bound_fixed}
    \end{equation}
    for any instance of the $\ppp$ problem.
\end{lemma}

\noindent Let $\eta_{n,p}^{q_n}$ denote the right-hand side of \eqref{eq:lower_bound_fixed}. The next lemma provides a useful monotonicity property. 

\begin{lemma}\label{lem:monotonicity_eta_lb_static}
    For any $n\geq 1$, let $q_n=\min\{ 1, 1/pn\}$. Then, we have $\eta_{n,p}^{q_n}\geq \eta_{n+1,p}^{q_{n+1}}$. Furthermore, $\lim_{n\to \infty} \eta_{n,p}^{q_n}=1-1/e$.
\end{lemma}

\noindent With these two lemmas, we are ready to provide the proof of Theorem \ref{thm:lb_for_quantile}. Indeed, for any $k\geq 1$, we obtain ${v(\ALG^{q_n})}/{v(\OPT)} \geq \eta_{n,p}^{q_n} \geq \eta_{n+k,p}^{q_{n+k}}$, where the first inequality simply follows from Lemma \ref{lem:general_formula}, and the second inequality follows from applying the monotonicity in Lemma~\ref{lem:monotonicity_eta_lb_static} iteratively $k$ times. Thus, we obtain $v(\ALG^{q_n})/v(\OPT)\geq \lim_{k \rightarrow \infty} \eta_{n+k,p}^{q_{n+k}}=1-1/e$ for any $n \geq 1$. We use the rest of this subsection to prove Lemmas~\ref{lem:general_formula} and~\ref{lem:monotonicity_eta_lb_static}.

\begin{proof}[Proof of Lemma~\ref{lem:general_formula}]
Fix $n \geq 1$, and we drop the index $n$ from quantile for simplicity. We start characterizing  $v(\ALG^q)$ as follows:
\begin{align}
    v&(\ALG^q)  = \sum_{a = 1}^n \Pr[\text{Bin}(n,q) = a] \left(\sum_{i = 1}^a p\cdot (1-p)^{i-1} \cdot (i-1) + a \cdot (1-p)^a\right) \cdot \E[X | X\geq \tau] \notag \\
    & = A_n(q, p) \frac{\int_\tau^\infty x \cdot f(x) \, \mathrm{d}x}{\Pr[X \geq \tau]} = A_n(q, p) \frac{\int_0^q F^{-1}(1-u) \, \mathrm{d}u}{q} = A_n(q, p) \frac{\int_0^1 r(v) \min\{v,q\} \, \mathrm{d}v}{q},\label{exp:threshold_val}
\end{align}
{ where the second equality uses $A_n(q, p) = (1-p) \cdot \left(1-(1-qp)^n\right)/p$, which is a simplification from the preceding expression. Note that $A_n(q,p)/q$ denotes the expected number of items accepted by the algorithm at quantile at most $q$, and $\E[X \mid X \geq \tau]$ represents the expected value of an item whose value lies above the threshold $\tau = F^{-1}(1-q)$. Hence, $v(\ALG)$ can be interpreted as \emph{the expected number of accepted items at quantile $q$ multiplied by the expected value of such an item};} the third equality uses the assumption that $F$ is strictly increasing and substitutes $F(x)$ by $1-u$, and the last equality is achieved by setting $F^{-1}(1 - u) = \int_u^1 r(v) \, \mathrm{d}v$, which is valid by our assumption on $F$. Then, per Proposition~\ref{prop:opt_value_complete}, we have

\begin{equation}\label{exp:comp_ratio_lb}
    \frac{v(\ALG^q)}{v(\OPT)} = \frac{\left( A_n(q, p) / q \right) \int_0^1 r(v) \min\{v,q\} \, \mathrm{d}v }{\int_{0}^1 B_n (p, v) r(v) \, \mathrm{d}v}
    \geq \inf_{v \in [0,1]} \frac{A_n(q, p)\cdot \min\{v,q\}}{B_n (p, v) \cdot q}.
\end{equation}
For $v\geq q$, the ratio inside the infimum becomes $A_n(p,q)/B_n(p,v)$ which is a decreasing function in $v$. For $v< q$, the ratio now becomes
\begin{align}\label{exp:comp_ratio_lb2}
    \frac{A_n(q, p) \cdot v}{B_n (p, v) \cdot q}  = \frac{A_n(q, p)}{q} \frac{v}{ \sum_{i=2}^n p\cdot (1-p)^{i-1} \sum_{j=1}^{i-1}\Pr[\text{Bin}(n,v) \geq j] + (1-p)^n \cdot nv}.  
\end{align}
In Lemma~\ref{lemma_mono}, we show that \eqref{exp:comp_ratio_lb2} is non-decreasing in $v$. This implies that the infimum is attained either when $v\to 0$ or  $v\to 1$. Thus, 
\begin{align*}
\frac{v(\ALG^q)}{v(\OPT)} &\geq \min\left\{ \lim_{v\to 0} \frac{A_n(q,p)v}{B_{n}(p,q)q} , \lim_{v\to 1} \frac{A_n(p,q)}{B_n(p,q)}  \right\} \\
&= \frac{(1-p)(1-(1-qp)^n)}{p}\cdot\min\left\{  \frac{1}{\E[\min\{D-1, n \} ]}, \frac{1}{(1-p)qn}  \right\},
\end{align*}
where the second line follows by a straightforward calculation of the corresponding limits. Here $D\sim \mathrm{Geom}(p)$, and a direct calculation shows that $\E[\min\{D-1,n\}]=(1-p)(1-(1-p)^n)/p$. This concludes the proof. 
\end{proof}

\begin{proof}[Proof of Lemma~\ref{lem:monotonicity_eta_lb_static}] We analyze three cases: (i) $p n < p(n+1)<1$; (ii) $p n \leq 1 < p(n+1) $; and (iii) $1< pn < p(n+1)$.

    \paragraph{Case (i).} Here, we have $q_{n}=q_{n+1}=1$. Then $\eta_{n,p}^{1} = \min\left\{  1,  \frac{1-(1-p)^n}{ pn}  \right\} = \frac{1-(1-p)^n}{pn}$,
    where the first equality follows simply by definition, while the second equality follows from the Bernoulli's inequality $(1-p)^n\geq 1-pn$. From here, it is immediate that $\eta_{n+1,p}^{1}\leq \eta_{n,p}^1$.

    \paragraph{Case (ii).} Here, we have $q_{n}=1$ and $q_{n+1}=1/p(n+1)$. Then $\eta_{n,p}^{q_n} = \frac{1-(1-p)^n}{pn}$,
    and
    \[
    \eta_{n+1,p}^{q_{n+1}}= \left(1-\left(1-\frac{1}{n+1}\right)^{n+1}\right) \min\left\{  1, \frac{1}{1-(1-p)^{n+1}}  \right\} = 1- \left( 1- \frac{1}{n+1} \right)^{n+1}
    \]
    with calculations analogous to the previous case. Note that the function $(1-(1-p)^n)/p =\sum_{\ell=0}^{n-1} (1-p)^\ell$ is decreasing in $p$. Thus, we have $p\in [1/(n+1),1/n]$, and
    \[
    \eta_{n,p}^{q_n} \geq 1- \left( 1-\frac{1}{n}\right)^n \geq 1- \left(1-\frac{1}{n+1} \right)^{n+1}= \eta_{n+1,p}^{q_{n+1}}.
    \]
    \paragraph{Case (iii).} In this case, we have $q_{n}=1/pn$ and $q_{n+1}=1/p(n+1)$. Thus, $\eta_{n,p}^{q_n}=1-\left( 1-\frac{1}{n}\right)^n$ and $\eta_{n+1,p}^{q_{n+1}}=1-\left( 1- \frac{1}{n+1} \right)^{n+1}$. As a byproduct of the analysis from the previous case, we have already proved $\eta_{n,p}^{q_{n}}\geq \eta_{n+1,p}^{q_{n+1}}$.
    To conclude, for $n>1/p$, we have $q_n=1/pn$ so that $\lim_{n \to \infty} \eta_{n,p}^{q_n} = \lim_{n \to \infty} 1- \left( 1- \frac{1}{n} \right)^n = 1- \frac{1}{e}$. 
\end{proof}

\subsection{Upper Bound on Competitive Ratio}\label{sebsec:comp_ratio_ub}
In this subsection, we prove $\gamma_{p, \zeta}^{\mathrm{NA}}\leq 1-1/e$ to complete the proof of Theorem~\ref{thm:static}. To do so, we provide a family of instances with competitive ratios approaching $1-1/e$. For this hard instance, we define the distribution through %{ $\Hat{F}(u)$:
% \begin{align}
%     F^{-1}(1 - u) = \frac{a_1}{n}\delta_{\lbrace 0 \rbrace}(u) + a_2\mathbbm{I}_{(0, \beta/n]}(u), \label{exp:hard_dist}
% \end{align}
\begin{align}
    \Hat{F}(u) := \frac{a_1}{n} \delta_{\lbrace 0 \rbrace}(u) + a_2\mathbbm{I}_{(0, \beta/n]}(u), % \approx \frac{a_1}{n} \left(\mathbbm{I}_{[0, 1/n^2]}(u) \cdot n^2 \right)+ a_2 \mathbbm{I}_{(1/n^2, \beta/n]}(u) ,
    \label{exp:hard_dist}
\end{align}
\noindent where $a_1, a_2, \beta \in \mathbb{R}$, $\beta \leq n$, $\delta_{\lbrace 0 \rbrace}(u)$ denotes the Dirac delta function centered at $u = 0$, and $\mathbbm{I}_S(\cdot)$ denotes the indicator of the set $S$. {  By approximating $\delta_{\{0 \}}(\cdot)$ via $n^2 \mathbbm{I}_{[0,1/n^2]}(\cdot)$, we can easily find a distribution $F$ such that $F^{-1}(1-u)\approx \hat{F}(u)$; however, we will provide the calculations over $\hat{F}$ as the calculations are cleaner. In short, $\hat{F}$ defines a three-point distribution: among $n$ draws, we expect approximately one high item with value of $a_1n$, $\beta$ moderate items with value of $a_2$, and the remaining items with value of $0$. With independent disruptions of probability $p$, the expected payoff of the $j^\text{th}$ accepted item is scaled by $(1-p)^j$; thus, taking many moderate items yields geometrically decaying value, while waiting for the single high item risks forfeiting most moderate items. We formalize this trade-off below.} Denote by $\text{Poisson}(\beta)$ a Poisson distribution with parameter $\beta$.  We use the following two lemmas to characterize the value of the optimal algorithm and the value of the single-threshold algorithm. 
\begin{lemma}{\label{lemma_hard_opt}} 
    When $n \to \infty$, $v(\OPT) \to a_1(1-p) + a_2 \sum_{j=1}^{\infty} \Pr[\text{Poisson}(\beta) \geq j] (1-p)^j$. 
\end{lemma}

\begin{lemma}\label{lemma_hard_threshold}
Let $\lambda=\lambda^* > 0$ be a solution to the equation
\begin{equation}
e^{-\lambda p} \left( a_1 + a_1 p \lambda + a_2 p \lambda^2 \right) = a_1. \label{eq:criticalpoint}
\end{equation}
As $n$ grows large, the expected value of $\ALG^{\lambda/n}$ converges to
\[
\lim_{n \to \infty} v(\ALG^{\lambda/n}) = \max \left\{ (1 - p)a_1,\; \frac{(1 - p)(1 - e^{-\beta p})}{\beta p} (a_1 + a_2 \beta),\; C_1 \right\},
\]
where $C_1=\frac{(1 - p)(1 - e^{-\lambda^* p})}{\lambda^* p} (a_1 + a_2 \lambda^*)$ if $\lambda^* \leq \beta$ and $C_1=0$ otherwise.
\end{lemma}

Let $f(a_1, a_2, \beta, p,\lambda):=v(\ALG^{\frac{\lambda}{n}})/v(\OPT)$. We aim to bound $
\min_{a_1, a_2, \beta, p} \max_{\lambda} f(a_1, a_2, \beta, p, \lambda)
$, which will provide an upper bound on the competitive ratio for non-adaptive algorithms.

\begin{theorem}{\label{hard_instance}} Given any $\varepsilon > 0$, for given inputs $a_1, a_2, \beta, p, \lambda$ which satisfy $a_2 = p (e-2) a_1$, and $\beta$ sufficiently large, we have $v(\ALG^{\frac{\lambda}{n}}) /v(\OPT) \leq 1-1/e +\varepsilon$. 
\end{theorem}
\begin{proof}[Proof of Theorem \ref{hard_instance}.]
    Given $\beta > 1/p$ sufficiently large, plugging in $a_2 = p (e-2) a_1$ into \eqref{eq:criticalpoint} we obtain the equation $e^{-\lambda^* p} \left( a_1 + a_1 p \lambda^* + p^2 (e-2) a_1 (\lambda^*)^2 \right) = a_1$, which holds if and only if $  \lambda^* = 1/p$ for $\lambda^*>0$.
    Note that when $\lambda^* = 1/p$, through direct comparison and monotonicity analysis, we get $$C_1 \geq \max \left \lbrace (1-p) a_1 , \frac{1-p}{\beta p}(1-e^{-\beta p}) (a_1 + a_2 \beta) \right\rbrace.$$ Consequently, for the input parameters $a_1, a_2, \beta, p, \lambda$ satisfying $a_2 = p (e-2) a_1$, and $\beta$ sufficiently large, the value $\lambda = 1/p$ serves as the maximizer of the function $f(a_1, a_2, \beta, p, \lambda)$. Then for any given $\varepsilon > 0$, we have 
    \begin{align*}
        \frac{v(\ALG^{\frac{1}{np}})}{v(\OPT)} & \leq \frac{(1-p) \cdot (1-1/e) \cdot (a_1+a_2/p)}{a_1(1-p) + a_2 \sum_{j=1}^{\infty} \Pr[\text{Poisson}(\beta) \geq j] (1-p)^j}. 
    \end{align*}
    Note that the desired upper bound of $1-1/e +\varepsilon$ is only achieved when the necessary condition of $\sum_{j=1}^{\infty} \Pr[\text{Poisson}(\beta) \geq j] (1-p)^{j-1} \geq 1/p -\varepsilon$ is met. We can find a sufficiently large $\beta$ that ensures the validity of the above inequality. The process of finding such a $\beta$ consists of the following two steps:
    \begin{enumerate}[label=(\roman*)]
    \item Find an integer $C_2$ such that $(1-p)^{C_2} \leq \varepsilon p /2$. $C_2$ could then be set as $\left\lceil \log_{1-p} \left( \varepsilon p /2\right) \right\rceil$. 
    \item Using the integer $C_2$ found in the previous step, we find $\beta^* \geq \lambda$ such that $\Pr[\text{Poisson}(\beta^*) \geq C_2] \geq 1-\varepsilon p /2$. The existence of such a $\beta^*$ is guaranteed by the intermediate value theorem.
    \end{enumerate}
    Using conditions (i) and (ii) above, we recover the necessary condition: 
    \begin{align*}
         \sum_{j=1}^{\infty} \Pr[\text{Poisson}(\beta) \geq j] & (1-p)^{j-1} \geq \sum_{j=1}^{C_2} \left(1-\frac{\varepsilon p}{2}\right) (1-p)^{j-1} = \left(1-\frac{\varepsilon p}{2}\right) \frac{1-(1-p)^{C_2}}{p}\\
        \geq & \left(1-\frac{\varepsilon p}{2}\right) \left(\frac{1}{p} - \frac{\varepsilon}{2}\right) = \frac{1}{p} - \varepsilon + \frac{\varepsilon^2 p}{4} \geq \frac{1}{p} - \varepsilon. \qedhere
    \end{align*}
\end{proof}

\iffalse
\subsubsection{Rare Disruption Regime} \label{subsub:small_p}

We briefly discuss the case of rare disruption when $p=\alpha/n$ for $\alpha\leq 1$. Using Lemma~\ref{lem:general_formula} with $q=1$, and using the inequality $1-x\leq e^{-x}$, we obtain the following lower bound on the competitive ratio:
\begin{align*}
    \left( \frac{1-e^{-\alpha}}{\alpha}\right)\min\left\{1, \frac{(1-\alpha/n)n}{\E[\min\{D-1,n\}]} \right\} 
\end{align*}
For $n$ tending to infinity, we obtain an asymptotic lower bound of $(1-e^{-\alpha})/\alpha$. This asymptotic competitive ratio larger than $ 1-1/e$ for $\alpha\in [0,1)$, improving the nonadaptive solution for fixed $p$. Furthermore, this function approaches $1$ when $\alpha\to 0$.
\fi

\section{The Class of Adaptive Algorithms}\label{sec:adapt}
%{  where $\zeta = 0$}

In this section, we focus on adaptive algorithms. In \S\ref{sec:adapt_lb}, we formally define our algorithm and present the proof of Theorem~\ref{thm:adaptive} to show that \textit{adaptivity} is key, and one can improve the competitive ratio from the non-adaptive case in the limit. In \S\ref{subsec:adapt_ub}, we present the proof of Theorem \ref{thm:adaptive_ub} to provide an upper bound on the competitive ratio, and we show that our derived asymptotic competitive ratio is tight. 

\subsection{Quantile-Based Threshold Algorithm}\label{sec:adapt_lb}

For each $i=1,\ldots,n$, our adaptive algorithm samples a quantile $q_i \in [0,1]$ from a density function with a support in $[\varepsilon_{i-1},\varepsilon_{i}]$, where $0=\varepsilon_0< \varepsilon_1< \cdots <  \varepsilon_n =1$. Then, if the observed value $X_i$ is at least $F^{-1}(1-q_i)$, the algorithm selects the value; otherwise the value is rejected. Denote this algorithm by $\ALG^{\mathrm{AD}}$. 

The construction of the densities over each interval $[\varepsilon_{i-1}, \varepsilon_i]$ follows a similar approach to the one outlined in~\cite{correa2021posted,perez2025iid}. We take some time to explain the importance of these densities and the challenges involved in applying the method from~\cite{correa2021posted}. From the second characterization of $v(\OPT)$ in Proposition~\ref{prop:opt_value_complete}, the function $g_n^{(0)}(p, q)=(1-p)n(1-pq)^{n-1}$ is the derivative of $A_n(p,q)=\E[\min\{ \text{Bin}(n,q),D-1 \}]$, i.e., the derivative of the expected number of values $\OPT$ gets at quantile $q$. The goal of our analysis is to show that, for every $q \in [0,1]$, the derivative of the expected number of values at quantile $q$ that the algorithm accepts is at least $\theta_n \cdot g_n^{(0)}(p, q)$. If this condition holds, then we can guarantee that $v(\ALG^{\mathrm{AD}}) \geq \theta_n \cdot v(\OPT)$. This analysis is quite stringent, as it requires specifying a \emph{valid density for every $q \in [0,1]$}. However, the method from \cite{correa2021posted} applied to the \ppp probem only provides a density over the interval $[0, p]$, leaving $(1-p,1]$ unassigned. To address this limitation, we use a direct approach and construct densities that cover the interval $[0,1]$ completely in such a way that the competitive ratio of $\ALG^{\mathrm{AD}}$ converges to $\theta^*$. 

We now explain the density functions for $\ALG^{\mathrm{AD}}$. Let $0=\varepsilon_0<\varepsilon_1$. For $\theta_n>0$, consider the following function $\beta_{1,n}(p,q):=-\theta_n \mathbbm{I}_{[\varepsilon_{0},\varepsilon_1]}(q) (g_n^{(0)})'(p,q)/(1-p)$. Note that $\beta_{1,n}\geq 0$. If we want to sample $q_1$ from $\beta_{1,n}(p,\cdot)$, then we must have $1=\int_0^1 \beta_{1,n}(p,q) {\, \mathrm{d}q}$, which happens if and only if $\frac{1}{n \theta_n} = 1- (1-p \varepsilon_1)^{n-1}$.
From here, $\varepsilon_1$ is decreasing in $\theta_n$; thus, there is $\theta_n$ such that $\varepsilon_1\leq 1$. In general, let $\beta_{i,n}(p,q):=- \theta_n \mathbb{I}_{[\varepsilon_{i-1},\varepsilon_i]}(q)(g_n^{(0)})'(p,q)/(1-p)$ such that the following system is satisfied for $0=\varepsilon_1< \varepsilon_2<\cdots< \varepsilon_k\leq 1$, and for the largest $k$ possible:
\begin{align}
    \int_0^1 \beta_{1,n}(p,q) {\, \mathrm{d}q} & =1, \label{eq:beta_1} \\
    \int_0^1 \beta_{i+1,n}(p,q) {\, \mathrm{d}q} & = \int_0^1 \beta_{i,n}(p,q)(1-pq) {\, \mathrm{d}q} & \forall i< k.\label{eq:beta_k}
\end{align}
We already know that we can satisfy this system with $k=1$. We seek to satisfy this system for $k=n$ and $\varepsilon_n=1$. Then, our densities for $\ALG^{\mathrm{AD}}$ become $\beta_{i,n}(p,q)/\int_0^1\beta_{i,n}(p,q) \, \mathrm{d}q$ for all $i\in [n]$. 

\begin{lemma}\label{lem:unique_theta}
    There is a unique $\theta_n>0$ such that the system $\eqref{eq:beta_1}-\eqref{eq:beta_k}$ has a solution for $k=n$ and $\varepsilon_n=1$. 
\end{lemma}

We present the proof of Lemma \ref{lem:unique_theta} after establishing the following guarantee on the competitive ratio of $\ALG^{\mathrm{AD}}$.

\begin{theorem}\label{thm:adapt_lower_bound}
    Let $\theta_n>0$ as in Lemma~\ref{lem:unique_theta}. Using the densities $\beta_{i,n}/\int_0^1 \beta_{i,n}(p,q)\, \mathrm{d}q$ for $\ALG^{\mathrm{AD}}$ guarantees
    \begin{equation}\label{exp:alg_adap}
         \frac{v(\ALG^{\mathrm{AD}})}{v(\OPT)}\geq  \left( 1-\frac{(1-p)^{n-1}pn}{1-(1-p)^n} \right) \cdot\theta_n.
    \end{equation}
\end{theorem}

\iffalse
\begin{align*}
    v(\ALG^{\mathrm{AD}}) & = \sum_{i=1}^n \Pr[\text{reach time }i]\int_0^1 \frac{\beta_{i,n}(p,u)}{\int_0^1 \beta_{i,n}(p,q)\, \mathrm{d}q} (1-p)\int_{F^{-1}(1-u)}^\infty x\, \mathrm{d}F(x)\, \mathrm{d}u
\end{align*}
Inductively, we can show that 
\[
\Pr[\text{Reach time }i] = \begin{cases}
    1 & i=1\\
    \int_0^1 \beta_{i-1,n}(p,q)(1-pq)\, \mathrm{d}q & i>1
\end{cases}
\]
\fi

\begin{proof}[Proof of Theorem \ref{thm:adapt_lower_bound}.]
Let $r_i$ be the probability that the algorithm observes $X_i$, $i \in [n]$. Then, we have
\begin{align*}
    v(\ALG^{\mathrm{AD}}) & = \sum_{i=1}^n r_i\int_0^1 \frac{\beta_{i,n}(p,u)}{\int_0^1 \beta_{i,n}(p,q)\, \mathrm{d}q} (1-p)\int_{F^{-1}(1-u)}^\infty x\, \mathrm{d}F(x)\, \mathrm{d}u,
\end{align*}
{  where $\beta_{i,n}(p,u)\big/\!\int_0^1 \beta_{i,n}(p,q)\,\mathrm{d}q$ is the conditional adaptive density of the quantile used when observing $X_i$, and the inner integral $(1-p)\int_{F^{-1}(1-u)}^\infty x\,\mathrm{d}F(x)$ is the expected marginal contribution at that step given quantile $u$, with $(1-p)$ accounting for survival upon acceptance.} Via induction, we have $r_1=1$, and for $i>1$, we have $r_i=\int_0^1 \beta_{i-1,n}(p,q)(1-pq)\, {d}q$. Hence, using the system $\eqref{eq:beta_1}-\eqref{eq:beta_k}$, we obtain
\begin{align}
    v(\ALG^{\mathrm{AD}}) & = \sum_{i=1}^n \int_{0}^{1} \beta_{i, n} (p, u) \cdot (1-p) \int_{0}^q F^{-1} (1-q) \, \mathrm{d}q \, \mathrm{d}u \label{eq:firstline52}\\
    & = \theta_n \int_{0}^1 F^{-1}(1-q) g_n^{(0)}(p, q) \, \mathrm{d}q - \theta_n \int_{0}^{1} F^{-1} (1-q) \cdot n (1-p)^{n} \, \mathrm{d}q,\label{ratio_adap_res0}
\end{align}
where \eqref{eq:firstline52} follows from a change of variable and using $r_i=\int_0^1 \beta_{i,n}(p,q)\, {d}q$ for $i >1$.  Then, using a ratio comparison, we have
\begin{align}
    \frac{\int_{0}^{1} F^{-1} (1-q) g_n^{(0)}(p,1)\, \mathrm{d}q}{\int_{0}^1 F^{-1}(1-q) g_n^{(0)}(p,q)\mathrm{d}q} & \leq \frac{(1-p)^{n-1} p n}{1-(1-p)^n}\leq 1,
    \label{exp:ratio_compare}
\end{align}
and applying this bound in \eqref{ratio_adap_res0}, we obtain
\[
v(\ALG^{\mathrm{AD}}) \geq \theta_n\cdot  \left( 1- \frac{(1-p)^{n-1}pn}{1-(1-p)^n} \right)\int_0^1 F^{-1}(1-u) g_n^{(0)}(p,u)\, \mathrm{d}u = \theta_n \cdot \left( 1- \frac{(1-p)^{n-1}pn}{1-(1-p)^n} \right) v(\OPT),
\]
which concludes the proof.
\end{proof}

We now present the proof of Lemma~\ref{lem:unique_theta}. The idea is to generalize the monotonicity of $\varepsilon_1$ as a function of $\theta_n$ to $\varepsilon_i$ for all $i$. To this end, we first present an alternative characterization of the system $\eqref{eq:beta_1}-\eqref{eq:beta_k}$.

\begin{proof}[Proof of Lemma \ref{lem:unique_theta}]
    We start with an intermediate result. Given a fixed $k \leq n$,  we claim that for all $i\leq k$, the following recursion holds: 
    \begin{equation}
        g_n^{(0)}(p, \varepsilon_{i}) - g_n^{(0)}(p, \varepsilon_{i-1}) = -\frac{1-p}{\theta_n} - p \left(\varepsilon_{i-1} \cdot g_n^{(0)}(p, \varepsilon_{i-1})  - \int_{0}^{\varepsilon_{i-1}} g_n^{(0)}(p, q) \, \mathrm{d}q \right). \label{exp:recursion_journal2}
    \end{equation}
    We proceed by induction. The base case $i = 1$ can be verified easily. Assuming that \eqref{exp:recursion_journal2} holds for some $i<k$, by the fundamental theorem of calculus, we have
    \begin{align*}
        g_n^{(0)}(p, \varepsilon_{i+1}) -& g_n^{(0)}(p, \varepsilon_{i})  = \int_{\varepsilon_i}^{\varepsilon_{i+1}} (g_n^{(0)})'(p,q) \, \mathrm{d}q = \int_{\varepsilon_{i-1}}^{\varepsilon_{i}} (g_n^{(0)})'(p, q) (1-pq) \, \mathrm{d}q \tag{per \eqref{eq:beta_k}}\\
        & = \int_{\varepsilon_{i-1}}^{\varepsilon_{i}} (g_n^{(0)})'(p, q) \, \mathrm{d}q - p \int_{\varepsilon_{i-1}}^{\varepsilon_{i}} (g_n^{(0)})'(p, q) q \, \mathrm{d}q\\
        & = g_n^{(0)}(p, \varepsilon_{i}) - g_n^{(0)}(p, \varepsilon_{i-1}) - p \varepsilon_i \cdot g_n^{(0)}(p, \varepsilon_i) + p \varepsilon_{i-1} \cdot g_n^{(0)}(p, \varepsilon_{i-1}) + p \int_{\varepsilon_{i-1}}^{\varepsilon_i} g_n^{(0)}(p,q) \, \mathrm{d}q \\
        & = -\frac{1-p}{\theta_n} + p \int_{0}^{\varepsilon_{i}} g_n^{(0)}(p, q) \, \mathrm{d}q  - p \varepsilon_i \cdot g_n^{(0)}(p, \varepsilon_i),
    \end{align*}
    where the last equality comes from the induction hypothesis. 

Note that from~\eqref{exp:recursion_journal2}, we obtain that $\varepsilon_{i}<\varepsilon_{i+1}$ and $(g_n^{(0)})'(p,\varepsilon_{i+1})\varepsilon_{i+1}' = \frac{1-p}{\theta_n^2} + (g_n^{(0)})'(p,\varepsilon_i)(1-p\varepsilon_i) \varepsilon_i'$, where the derivative is with respect to $\theta_n$. From here, we see that all $\varepsilon_i$'s are decreasing in $\theta_n$. Thus, by making $\theta_n$ sufficiently large, we can sequentially define $\varepsilon_{k+1}, \varepsilon_{k+2}, \dots\leq 1$, until we reach $\varepsilon_n = 1$. This concludes the proof of Lemma~\ref{lem:unique_theta}.
\end{proof}

\paragraph{Asymptotic Analysis.} We now show that $\liminf_{n \to \infty} \gamma_{n, p, { 0}}^{\mathrm{AD}} \geq \theta^* \approx 0.745$. Let $f_n(p,\lambda) := g_n^{(0)}(p,\lambda/n)/(1-p)n$. Then,~\eqref{exp:recursion_journal2} becomes
\begin{align}
    \frac{f_n(p, \lambda_{i}) - f_n(p, \lambda_{i-1})}{1/n} & = -\frac{1}{\theta_n} - p \left(\lambda_{i-1} \cdot f_n(p, \lambda_i) - \int_{0}^{\lambda_{i-1}} f_n(p,w) \mathrm{d}w\right), \label{exp:asymptotic_recursion}
\end{align}
where $\lambda_i=n\varepsilon_i$ for all $i$. Note that $\theta_n \in [0,1]$. Then, there exists a subsequence that converges to some $\hat{\theta}\in [0,1]$. For simplicity, we abuse notation and denote  this subsequence by $\theta_n$ so that $\theta_n\to \hat{\theta}\in [0,1]$. Now, doing a linear piece-wise approximation of $\lambda_i$ via a function $\lambda_n(x)$ such that $\lambda_i=\lambda_n(i/n)$ and taking the limit in $n$ of \eqref{exp:asymptotic_recursion}, we obtain
\begin{align*}
    f(p, \lambda(x))' &= -\frac{1}{\hat{\theta}} - p \left( \lambda(x) \cdot f(p,\lambda(x)) - \int_{0}^{\lambda(x)} f(p, w) \mathrm{d}w \right)\\
    & = 1- \frac{1}{\hat{\theta}} - p \lambda(x) e^{-\lambda(x)p} - e^{-\lambda(x)p},
\end{align*}
where $f(p,\lambda)= \lim_{n\to \infty} f_n(p,\lambda)= e^{-\lambda p}$. Furthermore, we have conditions $\lambda(0)=0$ and $\lim_{x\to 1} \lambda(1)=+\infty$. Performing the change of variable $y(x)=e^{-\lambda(x) p}$, we obtain the following system:
\begin{align*}
    y'(x) &= 1-\frac{1}{\hat{\theta}} + (\ln y(x)-1)y(x), \\
    y(0) & = 1, \quad \lim_{x\to 1}y(x)=0.
\end{align*}
This is exactly the Hill and Kertz equation presented in \S\ref{sec: prelim}, which has a solution if and only if $\hat{\theta}=\theta^*\approx 0.745$, which yields that $\liminf_n \theta_n = \theta^*$. Finally per Theorem~\ref{thm:adapt_lower_bound}, for any $p \in (0,1)$, we have
$$\liminf_{n \to \infty} \gamma_{n, p, { 0}}^{\mathrm{AD}} \geq \liminf_{n \to \infty} \theta_n \cdot  \left( 1-\frac{(1-p)^{n-1}pn}{1-(1-p)^n} \right) = \theta^*.$$

\subsection{Upper Bound} \label{subsec:adapt_ub}

In this subsection, we prove Theorem~\ref{thm:adaptive_ub} to establish that the lower bound derived in the previous section is optimal. Fix $p, \varepsilon \in (0,1)$ and $n$ sufficiently large so that $n \geq \left\lceil -\frac{\log(y(1-\varepsilon))}{p} \right\rceil$. We define the following distribution through $\Tilde{F}(u)$, where $y(\cdot)$ is the solution to the Hill and Kertz equation: 
% $$\Tilde{F}(u) = \frac{\theta^*}{(1-p)n} \cdot \delta_{\lbrace 0 \rbrace}(u) - \frac{p}{1-p} \int_{y^{-1} (e^{-pnu})}^{1-\varepsilon} \frac{1}{y'(s)} ds \mathbbm{I}_{(0, 1]}(u).$$
{ \begin{align}
\Tilde{F}(u) 
&:= \frac{\theta^*}{(1-p)n}\,\delta_{\{0\}}(u) 
   - \frac{p}{1-p} 
   \int_{y^{-1}(e^{-pnu})}^{1-\varepsilon} \frac{1}{y'(s)}\,\mathrm{d}s\;\mathbbm{1}_{(0,1]}(u). \notag %\\[6pt]
%&\approx \frac{\theta^*}{(1-p)n}\left(\mathbbm{1}_{[0,1/n^2]}(u)\cdot n^2\right)
%   - \frac{p}{1-p} 
%   \int_{y^{-1}(e^{-pnu})}^{1-\varepsilon} \frac{1}{y'(s)}\,ds\;\mathbbm{1}_{(1/n^2,1]}(u).
\label{eq:TildeF}
\end{align}}
\noindent {  Similar to the hard instance presented in \S\ref{sebsec:comp_ratio_ub}, the distribution comprises a spike when $u$ is close to $0$, together with a monotonically decreasing continuous component. Consequently, the largest contributions occur for small $u$, and the quantile decays rapidly thereafter. The instance is designed to induce a tension between capturing more mass immediately and preserving continuation value in the presence of disruption–induced discounting. We analyze this trade-off through a continuous Bellman equation for any online optimal algorithm, where the optimal control equates marginal benefit and discounted continuation cost. What follows is the formalization of this intuition. 
} 

Given $\varepsilon > 0$, denote the value of the optimal algorithm by $v^\varepsilon(\OPT)$ under $\Tilde{F}(u)$ that is defined above. The following proposition provides a characterization of $v^\varepsilon(\OPT)$.

\begin{proposition} \label{prop:ub_opt}
    $v^\varepsilon(\OPT) = \theta^* - \int_{0}^{1-\varepsilon} \frac{1}{y'(s)} \left( 1- \left(1+\frac{\log(y(s))}{n}\right)^n\right) \mathrm{d}s$. 
\end{proposition}

\noindent We then characterize the maximum value that can be obtained by any online algorithm. Consider the following dynamic program, where for any $n\geq1$ and $p\in(0,1)$, $D_i^\varepsilon$ is the value at observing $X_i$,\;$i \in [n+1]$, with the convention $X_{n+1}=0$:

\begin{equation} \label{exp:dp_ub}
    D_{i}^\varepsilon = \sup_{q \in [0,1]} \left \lbrace (1-p) \int_{0}^q \Tilde{F}(u) \, \mathrm{d}u + (1-pq)D_{i+1}^\varepsilon \right \rbrace, \forall i \in [n] \hspace{0.5cm} \text{and} \hspace{0.5cm} D_{n+1}^\varepsilon = 0.
\end{equation}

\noindent In particular, we are interested in analyzing the dynamic program solution $D_i^0$ for all $i \in [n]$. We now examine its continuous approximate counterpart (denoted by $d(x)$) through Lemma~\ref{lem: ub_adapt_be} and link it back to discrete valued $D_i^0$ via Lemma~\ref{lem: ub_dp_approx}. 
In order to find $d(0)$, we rewrite \eqref{exp:dp_ub} from an ODE perspective. Consider the following Bellman equation:
\begin{align}
    - d'(x) = \sup_{\mu \in [0,\infty]} \left\lbrace \int_{0}^{\mu} h(u) \, \mathrm{d}u -p \mu d(x) \right\rbrace, \label{exp: ub_adapt_be}
\end{align}

\noindent where $h(u):=  \theta^* \cdot \delta_{\lbrace 0 \rbrace}(u) - p \int_{y^{-1} (e^{-pu})}^{1} \left( 1/y'(s) \right) \mathrm{d}s \mathbbm{I}_{(0, 1)}(u)$. In what follows, we show that for any $x \in (0, 1]$, \eqref{exp: ub_adapt_be}  can be satisfied by 
\begin{align}
    d(x) = \int_{x}^1 -\frac{1}{y'(s)} \mathrm{d}s \hspace{0.5cm} \text{and} \hspace{0.5cm} \mu = -\frac{\log(y(x))}{p} \label{exp: ub_adapt_sol}.
\end{align} 

\begin{lemma}\label{lem: ub_adapt_be}
    \eqref{exp: ub_adapt_sol} provides a unique solution to the Bellman equation \eqref{exp: ub_adapt_be}, when $x \in (0, 1]$. 
\end{lemma}

Next, we connect the continuous dynamic program value (denoted by $d(i/n)$) with $D_i^\varepsilon$. 

\begin{lemma}\label{lem: ub_dp_approx}
    Fix $\sigma \in (0,1)$ and $n > -\log(y(1-\sigma))$. Let 
\begin{equation} \label{exp: eta_req}
    \eta_\sigma := \frac{n\sigma - \log(y(1-\sigma))(1 - \sigma)}{(1 - \sigma)\left(n + \log(y(1-\sigma))\right)}. 
\end{equation}
Let $\widetilde{D}_i := (1+ \eta_\sigma) \cdot d((1-\sigma)i/n))$. Then for any $i \in [n]$, we have $D_i^\varepsilon \leq D_i^0 \leq \widetilde{D}_i$. 
\end{lemma}

\begin{proof}[Proof of Lemma \ref{lem: ub_dp_approx}.]
    The first inequality follows from observing that $\Tilde{F}(u)$ is a positive, non-increasing function in $\varepsilon$. Thus, per \eqref{exp:dp_ub}, it holds that $D_i^\varepsilon \leq D_i^0$ for any $\varepsilon> 0$. To prove the second inequality, we use the following claim. 
    \begin{claim} \label{clm: ub_tilde}
        For any $q \in [0,1]$, we have $(1-p) \int_{0}^q \Tilde{F}(u) \, \mathrm{d}u + (1-pq)\widetilde{D}_{i+1} \leq \widetilde{D}_{i}$. 
    \end{claim}
    Now we are ready to show the second inequality. We proceed by induction. Note that for some $q_i \in [0,1]$, the Bellman equation for $D_i^0$ is satisfied with equality: 
    \begin{align*}
        D_{i}^0 &=  (1-p) \int_{0}^{q_i} \Tilde{F}(u) \, \mathrm{d}u + (1-pq_i)D_{i+1}^0 \\
        & \leq (1-p) \int_{0}^{q_i} \Tilde{F}(u) \, \mathrm{d}u + (1-pq_i) \widetilde{D}_{i+1}^0 \leq \widetilde{D}_{i} \qquad \text{(per induction hypothesis and Claim~\ref{clm: ub_tilde})}%\qedhere
    \end{align*}
\end{proof}

Finally, we compare the value of the dynamic program and the optimal algorithm. 
Specifically, we consider the ratio $D^\varepsilon_1/ v^\varepsilon(\OPT)$, and analyze its behavior in the asymptotic regime. Noting that $D_1^\varepsilon \leq D_1^0 \leq (1+\eta_\sigma) d((1-\sigma)/n)$ and $\eta_\sigma\to \sigma/(1-\sigma)$ when $n\to \infty$, we have
\begin{align*} \label{exp:compare_ub}
    \lim_{n\to \infty}\frac{D^\varepsilon_1}{v^\varepsilon(\OPT)} & \leq   \lim_{n\to \infty} \frac{(1+ \eta_\sigma) \cdot d((1-\sigma)/n))}{v^\varepsilon(\OPT)}  =\frac{1}{1-\sigma}\frac{\int_{0}^{1} -\frac{1}{y'(s)} ds}{\theta^* - \int_{0}^{1-\varepsilon} \frac{1-y(s)}{y'(s)}  ds},
\end{align*}
where in the last inequality we use Proposition \ref{prop:ub_opt} and Lebesque's dominated convergence theorem. This bound holds for any $\sigma\in (0,1)$. Thus,
\[
\gamma_{p, { 0}}^{\mathrm{AD}} \leq \frac{\int_{0}^{1} -\frac{1}{y'(s)} ds}{\theta^* - \int_{0}^{1-\varepsilon} \frac{1-y(s)}{y'(s)}  ds}
\]
for any $\varepsilon \in (0,1)$. Then letting $\varepsilon\to 0$, we see that the right-hand side tends to $\theta^*$ per \cite[Theorem 3.11]{liu2020variable}, which concludes the proof of Theorem \ref{thm:adaptive_ub}.

{ 
\section{Partial Recovery} \label{sec:recovery}

In this section, we examine the \ppp problem under partial recovery with parameter $\zeta \in [0,1]$. The proofs of Theorems~\ref{thm:static} and~\ref{thm:adaptive} turn out to be analogous to those in the case $\zeta=0$. For completeness, we present the main steps of the arguments, while omitting the full details to avoid redundancy.

%In particular, we show that when we are allowed to recover a fraction $\zeta$ of the last selected value that caused a disruption, the competitive ratios of both algorithms remain unchanged. 

% In this section, we complete the proofs of Theorems~\ref{thm:static} and~\ref{thm:adaptive} by allowing partial recovery of the last accepted value after a disruption. Assuming a fixed partial recovery $\zeta$, we show that the competitive ratios of both algorithms are unchanged. Specifically, suppose that the last selected item yields a linear recovery equal to $\zeta$ times its original value, where $\zeta \in [0,1]$. Under the partial-retention modification, the value expressions for the non-adaptive algorithm $v(\ALG^q)$, the adaptive algorithm $v(\ALG^{\mathrm{AD}})$, and the corresponding clairvoyant benchmark $v(\OPT)$ take similar form as in \S\ref{subsec:static_cr} and \S\ref{sec:adapt_lb}, with only the disruption payoff altered. %We provide the details below. 

\subsection{Non-adaptive Algorithm} \label{sec:recovery_na}
We first characterize the value of the non-adaptive algorithm similar to Lemma~\ref{lem:general_formula}. 
\begin{proposition}\label{prop: v_non_adapt}
We have $v(\ALG^q)= \big( A_n(q,p) + \zeta \big(1 - (1-qp)^n \big) \big)\,\left( \int_0^1 r(v) \min\{v,q\} \mathrm{d}v \right) / q. $
\end{proposition}

\begin{proof}[Proof of Proposition \ref{prop: v_non_adapt}.]
Relative to \eqref{exp:threshold_val}, the only change is that a disruption upon the $i$-th acceptance yields a fraction $\zeta$ of the last item; hence, an additive term $\zeta$ is included when a disruption occurs: 
\begin{align}
v(\ALG^q)
&= \sum_{a=1}^n \Pr[\mathrm{Bin}(n,q)=a]\;
\Bigg(\sum_{i=1}^{a} p(1-p)^{i-1}\big[(i-1)+\zeta\big] + a(1-p)^{a}\Bigg)\;
\mathbb{E}[X\mid X\ge \tau] \notag\\
&= \Big( A_n(q,p) + \zeta\ \sum_{a=0}^n \Pr[\mathrm{Bin}(n,q)=a]\,(1-(1-p)^a) \Big)
\frac{\int_\tau^\infty x f(x) \mathrm{d}x}{\Pr (X\ge \tau)} \notag\\
&= \Big( A_n(q,p) + \zeta\,[1-(1-qp)^n] \Big)\;
\frac{\int_0^1 r(v)\,\min\{v,q\} \mathrm{d}v}{q}, \notag
\end{align}
where the last equality follows by the same procedure as in the proof of Lemma~\ref{lem:general_formula}.
\end{proof}
% For any realization $a$ of $\mathrm{Bin}(n,q)$, suppose partial recovery $\zeta \in[0,1]$ is awarded on the disrupting last item. In that case, each disrupted outcome contributes $(i-1)+\zeta$, while the no–disruption outcome contributes $a$. Hence the expected number of items collected by the single–threshold algorithm is
% \[
% \sum_{i=1}^{a} p(1-p)^{i-1}\big[(i-1)+\zeta\big] + a(1-p)^{a}
% = \left(\sum_{i=1}^{a} p(1-p)^{i-1}(i-1) + a(1-p)^{a}\right) + \zeta \left(1-(1-p)^{a} \right).
% \]
% While for the optimal algorithm, the marginal contribution from partial recovery equals \textcolor{red}{I am confused why $\zeta$ is outside?}
% \[
% \zeta \sum_{i=1}^{n} p(1-p)^{i-1} \mathbb{E} \left[X_{(i)}\right],
% \qquad
% \text{where }\;
% \mathbb{E}\!\left[X_{(i)}\right]
% = \int_{0}^{1} \Pr\!\big[\mathrm{Bin}(n,v)\ge i\big]\; r(v)\, dv .
% \]
\noindent Combining Propositions~\ref{prop:opt_value_complete} and~\ref{prop: v_non_adapt} yields the following ratio between $v(\ALG^q)$ and $v(\OPT)$: 
\begin{align*}
    \frac{v(\ALG^q)}{v(\OPT)}
= \frac{\big( A_n(q,p) + \zeta \big(1 - (1-qp)^n \big) \big)\,\left( \int_0^1 r(v) \min\{v,q\} \mathrm{d}v \right) / q}
       {\int_{0}^{1} \left(B_n(p,v) + \zeta \big[1 - (1 - vp)^n\big] \right) r(v) \mathrm{d}v}.
% \label{exp:partial_rec_non}
\end{align*}
The same monotonicity arguments as in the proof of Lemma \ref{lem:general_formula} imply that the infimum is attained when $v\to 0$ or $v\to 1$. Consequently,
\begin{align*}
& \frac{v(\ALG^q)}{v(\OPT)} \geq \min\left\{ \lim_{v\to 0} \frac{\big( A_n(q,p) + \zeta\big(1 - (1-qp)^n \big) \big) v}{\left(B_n(p,v) + \zeta \big[1 - (1 - vp)^n\big] \right) q} , \lim_{v\to 1} \frac{\big( A_n(q,p) + \zeta\big(1 - (1-qp)^n \big) \big)}{\left(B_n(p,v) + \zeta \big[1 - (1 - vp)^n\big] \right) }  \right\} \\
&= \frac{(1-p+p\zeta)(1-(1-qp)^n)}{p}\cdot\min\left\{  \frac{1}{[(1-p)/p +\zeta] \left(1 - (1-p)^n \right)}, \frac{1}{(1-p + \zeta p)qn}  \right\},
\end{align*}
This matches Lemma~\ref{lem:general_formula} by canceling common factors that appear both in numerator and denominator. Similar to Lemma~\ref{lem:monotonicity_eta_lb_static}, the same $1-1/e$ bound holds when no recovery is allowed. 
%M\sk{is it better to state a similar to Lemma 4.3 here to conclude? I think it is more formal}

%{\color{red} How about the upper bound? Does the same distribution or approach work here? If so, make a remark}

\begin{remark}
For the upper bound, we may reuse the distribution we defined for a hard instance from \S\ref{sebsec:comp_ratio_ub}. Since both $v(\ALG)$ and $v(\OPT)$ are scaled by the same factor $\big( (1-p+p\zeta)/(1-p) \big)$, the ratio is preserved, and the remainder of the analysis follows.
\end{remark}

\subsection{Adaptive Algorithm} \label{sec:recovery_a}
We first characterize the value of the adaptive algorithm in the partial recovery case, in parallel to Theorem~\ref{thm:adapt_lower_bound}. 
Recall that $g_{n}^{(\zeta)}(p,q) = n\bigl(1-(1-\zeta)p\bigr)\,(1-pq)^{n-1}$. Define the adaptive scheme via
$\beta^{(\zeta)}_{i,n}(p,q):=- \theta_n \mathbb{I}_{[\varepsilon_{i-1},\varepsilon_i]}(q)(g_n^{(\zeta)})'(p,q)/(1-(1-\zeta)p)$,
where $0=\varepsilon_1<\varepsilon_2<\cdots<\varepsilon_k\le 1$ is chosen for the largest $k$ possible so that
\[
\int_0^1 \beta^{(\zeta)}_{1,n}(p,q)\,\mathrm{d}q=1,
\qquad
\int_0^1 \beta^{(\zeta)}_{i+1,n}(p,q)\,\mathrm{d}q
=
\int_0^1 \beta^{(\zeta)}_{i,n}(p,q)\,(1-pq)\,\mathrm{d}q, 
\quad \forall i<k.
\]
Using the expression for $v(\OPT)$ from Proposition~\ref{prop:opt_value_complete}, we obtain the following:
\begin{proposition}$v(\ALG^{\mathrm{AD}})
\;=\;
\theta_n \, v(\OPT)
\;-\;
\theta_n \int_{0}^{1} F^{-1}(1-q)\,
n\bigl(1-(1-\zeta)p\bigr)\,(1-p)^{\,n-1}\,\mathrm{d}q.$ \label{prop:adaptivewithrecovery}
\end{proposition}
\noindent The proof is straightforward with 
$v(\ALG^{\mathrm{AD}})
=
\sum_{i=1}^n
\int_{0}^{1}
\beta^{(\zeta)}_{i,n}(p,u)\,(1-p)
\!\int_{0}^{u} F^{-1}(1-q)\,\mathrm{d}q\,\mathrm{d}u$ and using the system of recursions above. Substituting $v(\OPT)=\int_0^1 F^{-1}(1-q)\,g_n^{(\zeta)}(p,q)\,\mathrm{d}q$ yields Proposition \ref{prop:adaptivewithrecovery}. Thus, the ratio between $v(\ALG^{\mathrm{AD}})$ and $v(\OPT)$ can be written as 
% Similar to the non-adaptive algorithm, we rewrite $v(\ALG^{\mathrm{AD}})$, and the corresponding $v(\OPT)$. By retrieving partial recovery from the last accepted item, and adopting the short-handed notation -- $g_{n}^{(\zeta)}(p,q) = n\bigl(1-(1-\zeta) p\bigr) (1-pq)^{n-1}$, and setting corresponding density function $\beta_{i,n}(p,q)=- \theta_n \mathbb{I}_{[\varepsilon_{i-1},\varepsilon_i]}(q)g_n^{(\zeta)'}(p,q)/(1-(1-\zeta)p)$, we yield the reformulated ratios below, where in particular, $\zeta=0$ recovers the zero recovery case from \S\ref{sec:adapt}: 
% \begin{align}
%     \frac{v(\ALG^q)}{v(\OPT)} = \frac{\left( A_n(q, p) + k\left(1 - (1-qp)^n \right) \right) \left( \int_0^1 r(v) \min\{v,q\} \, \mathrm{d}v \right) / q}{\int_{0}^{1} \left(B_n(p,v) + k \big[1 - (1 - vp)^n\big] \right) r(v)  \, \mathrm{d}v} \label{exp:partial_rec_non} \\
%     \frac{v(\ALG^{\mathrm{AD}})}{v(\OPT)} = \frac{\theta_n \int_{0}^{1} F^{-1}(1-q)\, g_n^{(k)}(p,q)\, \mathrm{d}q
% \;-\; \theta_n \int_{0}^{1} F^{-1}(1-q)\, n(1-p)^n\, \mathrm{d}q}{\int_{0}^{1} F^{-1}(1-q)\; g_n^{(k)}(p,q) \mathrm{d}q} \label{exp:partial_rec_ad}
% \end{align}
\begin{align}
\frac{v(\ALG^{\mathrm{AD}})}{v(\OPT)}
= \frac{\theta_n \int_{0}^{1} F^{-1}(1-q) g_n^{(\zeta)}(p,q)\, \mathrm{d}q
\;-\; \theta_n \int_{0}^{1} F^{-1}(1-q)\, n \bigl(1-(1-\zeta) p\bigr) (1-p)^{n-1} \mathrm{d}q}
       {\int_{0}^{1} F^{-1}(1-q) g_n^{(\zeta)}(p,q) \mathrm{d}q}.
\label{exp:partial_rec_ad}
\end{align}
The conclusion of Theorem~\ref{thm:adapt_lower_bound} continues to hold, since the ratio comparison in~\eqref{exp:ratio_compare} is unchanged after canceling out $\bigl(1-(1-\zeta)p\bigr)$. Note that the connection still holds, where $g_{n}^{(\zeta)}(p,q)$ is the derivative of $\big( A_n(q,p) + \zeta \big(1 - (1-qp)^n \big) \big)$, i.e.,\ the derivative of the expected number of values $\OPT$ gets at quantile at most $q$ plus a fraction $\zeta$ recovery. From here, we adopt the same analysis from the base model, and we observe that the competitive ratio expression is unaffected. The proof is analogous to the base case and is therefore omitted for brevity. 

% {\color{red} SAme here regarding upper bound.}

\begin{remark} For the upper bound, we redefine the distribution $\acute{F}(u)$ in \S \ref{subsec:adapt_ub} as: 
\begin{align}
\Tilde{F}(u) 
&:= \frac{\theta^*}{(1-p+p\zeta)n}\,\delta_{\{0\}}(u) 
   - \frac{p}{1-p+p\zeta} 
   \int_{y^{-1}(e^{-pnu})}^{1-\varepsilon} \frac{1}{y'(s)}\,\mathrm{d}s\;\mathbbm{1}_{(0,1]}(u). \notag %\\[6pt]
%&\approx \frac{\theta^*}{(1-p)n}\left(\mathbbm{1}_{[0,1/n^2]}(u)\cdot n^2\right)
%   - \frac{p}{1-p} 
%   \int_{y^{-1}(e^{-pnu})}^{1-\varepsilon} \frac{1}{y'(s)}\,ds\;\mathbbm{1}_{(1/n^2,1]}(u).
% \label{eq:TildeFnew}
\end{align}
Under this definition and the new partial-recovery dynamics, the adaptive and optimal algorithms attain the same value, and the remainder of the analysis follows unchanged. 
\end{remark}

%Furthermore, we demonstrate that the bound is tight. In \S\ref{subsec:objective_max}, we modify the objective from the sum of selected values to the maximum selected value. Under this objective, a single-threshold policy guarantees a competitive ratio of at least $1-1/e$. Moreover, the guarantee varies monotonically with $p$, decreasing as $p$ increases, in contrast to the sum objective where the ratio is $p$-independent. 
%\section{Extension (This sounds more like a variation of OSUD. Find another name for the 
%section)}\label{sec:extensions}

\section{Maximizing the Largest Value Until Disruption}\label{sec:extensions}

In this section, we consider the following variant of the \ppp{} problem: if $S$ represents the set of indices of values selected by $\ALG$ before a disruption, then the value of the algorithm $\ALG$ is $v_{\max}(\ALG)=\E[\max_{i\in S}\{ X_i \}]$. We compare the value of an algorithm against the optimal clairvoyant algorithm $\OPT$ that knows the values upfront but not the disruption; hence, it is immediate that $v_{\max}(\OPT)=(1-p)\E[X_{(1)}]$. Our main result is a constant competitive ratio, summarized in the following result.

\begin{theorem}\label{thm:lowerbound_max}
    For any $n\geq 1$, $\sup_{\ALG}\inf_{F \in \mathcal{F}} \frac{v_{\max}(\ALG) }{v_{\max}(\OPT)}\geq \max_{\lambda\leq n} \min\left\{ \frac{1-e^{-p\lambda}}{p \lambda} , 1-e^{-\lambda}\right\} \geq 1-1/e $.
\end{theorem}

To prove this result, we utilize a non-adaptive single-threshold algorithm that is quantile-based. Similar to our analysis of NA in~\S\ref{sec:static}, we can provide an exact formula for $v_{\max}(\ALG^q)$ where $\ALG^q$ is an algorithm that receives a quantile $q\in [0,1]$ and uses the threshold $\tau$ such that $q=\Pr(X\geq \tau)$. We provide a detailed proof in~\S\ref{subsec:proof_max}.

The lower bound $\max_{\lambda \le n} \min \left\{ \frac{1 - e^{-p\lambda}}{p\lambda}, 1 - e^{-\lambda} \right\}$ converges to $1 - e^{-\lambda(p)}$ as $n\to\infty$, where $\lambda = \lambda(p)$ is the unique solution to $1 - e^{-p\lambda} = p\lambda (1 - e^{-\lambda})$. The following result shows that this lower bound is essentially tight for the class of non-adaptive algorithms.

\begin{proposition}
    For any non-adaptive algorithm $\ALG$, $\inf_{F \in \mathcal{F}, n\geq 1} \frac{v_{\max}(\ALG)}{v_{\max}(\OPT)} \leq 1-e^{-\lambda(p)}$.
\end{proposition}

To prove this result, we provide a hard instance similar to the one provided for the class NA in~\S\ref{sec:static}. We defer the details to Appendix~\ref{subsec:max_largest_ub}. In Figure~\ref{fig:max_comp_ratio}, we present the plot of $1-e^{-\lambda(p)}$. We note that, numerically, it is a decreasing function of $p$.

\begin{figure}[H]
  \centering
  % scale by width (recommended); use .pdf/.png/.jpg
  \includegraphics[width=0.5\linewidth]{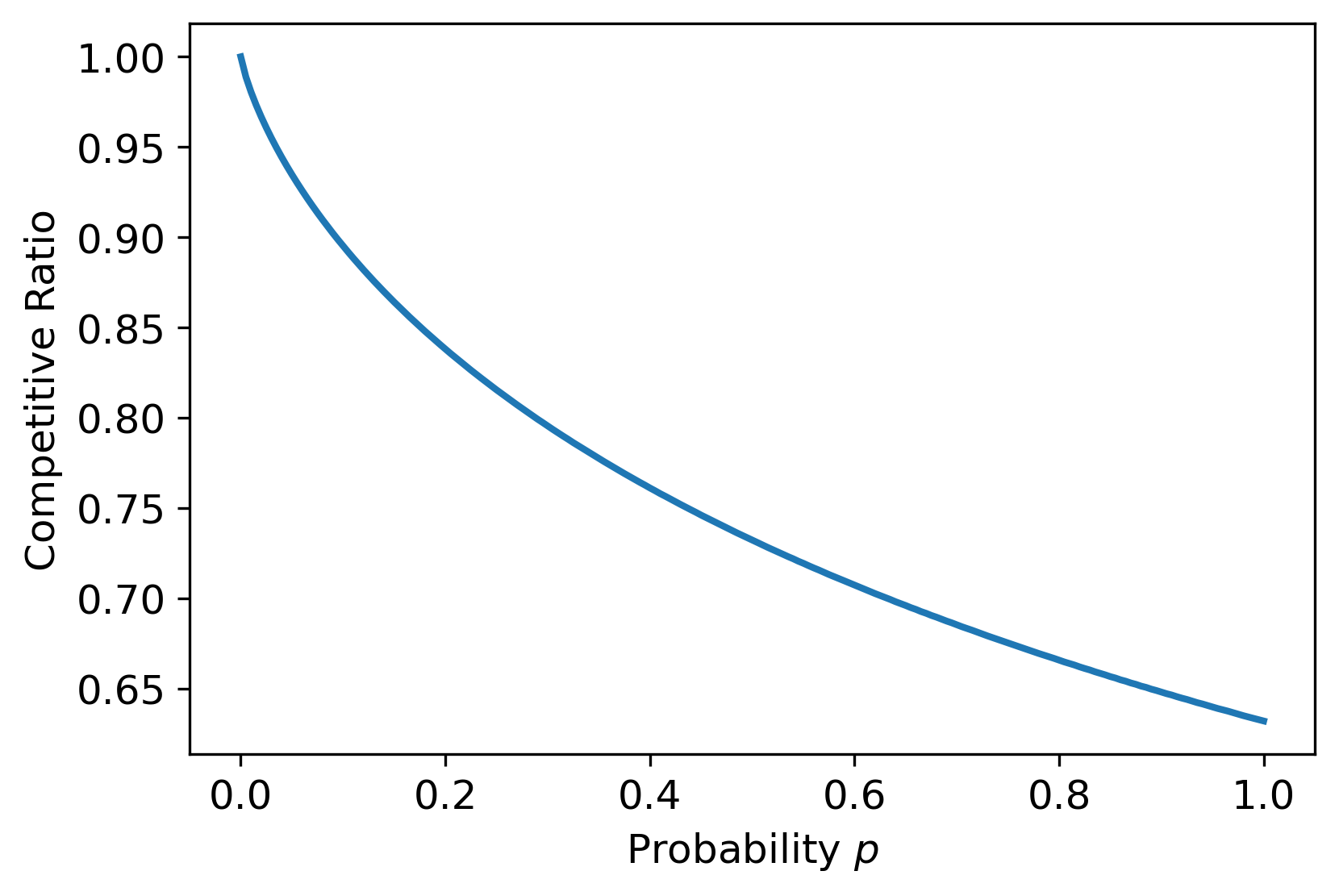}
  \caption{Competitive ratio lower bound computed from~\eqref{exp:cr_largest} for $p \in [0,1)$.}
  \label{fig:max_comp_ratio}
\end{figure}

\begin{remark}
    We can consider an alternative model where $S$ is the set of all indices of values selected by $\ALG$ including the index where the disruption occurs. Our results remain valid for $p<1$. We omit the details for brevity.
\end{remark}

\begin{remark}
We note that $\sup_{\ALG} v_{\max}(\ALG)$ can be computed via dynamic programming. However, any natural formulation must track both the current retained maximum $x$ and the remaining number of items $i$. The optimal policy takes the form of a threshold policy $\tau_{i,x}$, yet characterizing or even approximating these thresholds to allow a competitive analysis remains a challenging problem, which we leave as an open question.  A similar challenge is faced in prophet inequalities, where multiple items can be selected (e.g., see~\cite{assaf2000simple,harb2025new}). 
\end{remark}

\subsection{Proof of Theorem~\ref{thm:lowerbound_max}}\label{subsec:proof_max}

    We use following two lemmas to characterize the expected values of the optimal algorithm and the single-threshold algorithm using the function $r(v)>0$, where $\int_{u}^1 r(v)\mathrm{d}v = F^{-1} (1-u)$.  

    \begin{lemma} \label{lem:opt_max}
       We have $v_{\max}(\OPT) = (1-p) \int_{0}^{1} (1-(1-v)^n) \cdot r(v) \mathrm{d}v$.
    \end{lemma}
    The optimal policy simply targets the largest value. With probability \(1-p\) the item is not disrupting, and the expected collected value equals \(\mathbb{E}[X_{(1)}]\); hence Lemma~\ref{lem:opt_max} follows immediately from Claim~\ref{lem:ordered_stat}.
    
    \begin{lemma} \label{lem:single_max}
        We have $v_{\max}(\ALG^q) = \sum_{a=1}^n \Pr[\mathrm{Bin}(n,q)=a]\!\left(\sum_{k=0}^{a-1} p(1-p)^k\,\mu_k \;+\; (1-p)^a\,\mu_a\right)$\text{, where} $\mu_l := \E[\max\{X_1,\dots,X_l\}| X_1\geq \tau,\dots,X_l\geq \tau] = \left( 1/q^{l} \right) \int_{0}^{1} \left( q^l-(q-\min \{q,v\})^l \right) r(v) \mathrm{d}v$. 
    \end{lemma}
     The proof follows a similar analysis to Lemma~\ref{lem:general_formula}, but the gain from an accepted item depends on the realized disruption horizon; instead of obtaining \(\mathbb{E}[X \mid X\ge \tau]\) from each acceptance, the expected value is \(\mathbb{E} \left[\max\{X_1,\ldots,X_l\} \mid X_1\ge \tau,\ldots,X_l\ge \tau\right]\), where \(l\) is the number of selections survived without a disruption. Full proof for Lemma~\ref{lem:single_max} is provided in \S \ref{app:sec3}. Let
     % \begin{align*}
     %     W_{\ALG}(v) :=\sum_{a=1}^{n}\Pr[X=a] \left(\sum_{k=0}^{a-1}p\,(1-p)^{k} \frac{q^k-\bigl(q-\min\{q,v\}\bigr)^k}{q^k}+(1-p)^{a} \frac{q^a-\bigl(q-\min\{q,v\}\bigr)^a}{q^a}\right)
     % \end{align*}

     \[
    W_{\ALG}(v):=
    \begin{cases}
    (1-p)\Bigl[1-(1-q)^{n}\Bigr], 
    & \text{if } v\ge q,\\[6pt]
    \dfrac{(1-p)\,(v/q)}{\,p+(1-p)\,(v/q)\,}\;\Bigl[1-\bigl(1-qp-(1-p)\,v\bigr)^{n}\Bigr], 
    & \text{if } 0<v<q.
    \end{cases}
    \]

     %Define $W_k^q(v)$, $T(a,v)$, and $W_{\ALG}(v)$ as follows, we can also write $v_{\max}(\ALG^q)$ in terms of $W_{\ALG}(v)$: 
% {\color{red} Do you need all these auxiliary functions? They barely appear in this proof. Make the presentation simpler in the main body.}
     
    % \begin{align}
    % W_k^q(v) := &\frac{q^k-\bigl(q-\min\{q,v\}\bigr)^k}{q^k} \quad ;\quad T(a;v) := \left(\sum_{k=0}^{a-1}p\,(1-p)^{k}\,W_k^q(v)+(1-p)^{a}\,W_a^q(v)\right), \notag \\
    % & W_{\ALG}(v) :=\sum_{a=1}^{n}\Pr[X=a] \cdot T(a;v) \quad ;\quad v_{\max}(\ALG^q) := \int_{0}^1 W_{\ALG}(v) \cdot r(v) \mathrm{d}v. \label{exp:max_alg_q}
    % \end{align}

\noindent We reorganize the expression given in Lemma~\ref{lem:single_max} into a single integral and obtain a representation of $v_{\max}(\ALG^q)$ using the shorthand notation $W_{\ALG}(v)$: 
 $$v_{\max}(\ALG^q) = \int_{0}^{1} W_{\ALG}(v)\, r(v)\, \mathrm{d}v.$$ 
%    We now characterize $W_{\ALG}(v)$ for different values of $v$:
%    % We can find expressions for different $v$ value regimes through direct calculation: 
%     \paragraph{Case \(v\geq q\).} 
%     % For any \(a\geq 1\), $T(a;v) = 1-p$. Therefore, 
%     $W_{\ALG}(v)=(1-p)\,\bigl[\,1-(1-q)^{n}\,\bigr].$
%     \paragraph{Case \(0<v<q\).} 
%     %For any \(a\geq 1\),
%     %$$T(a;v) = \frac{(1-p)\,(v/q)}{\,p+(1-p)\,(v/q)\,}\;\bigl[1-((1-p)(1-v/q))^{a}\bigr].$$
%     %Therefore, 
%     $W_{\ALG}(v)\;=\;\frac{(1-p)\,(v/q)}{\,p+(1-p)\,(v/q)\,}\; \Bigl[\,1-\bigl(1-qp-(1-p)\,v\bigr)^{n}\Bigr].$
%     \vspace{0.25cm}

   \noindent Finally, we compare  $v_{\max}(\ALG^q)$ and $v_{\max}(\OPT)$: 
    \begin{align}
    & \frac{v_{\max}(\ALG^q)}{v_{\max}(\OPT)}  = \frac{ \int_{0}^1 W_{\ALG}(v) \cdot r(v) \mathrm{d}v}{(1-p) \int_{0}^{1} (1-(1-v)^n) \cdot r(v) \mathrm{d}v} \notag \\
    & \geq \min \left\{\inf_{v \in [0,q)} \frac{(1-p)\,(v/q)}{\,p+(1-p)\,(v/q)\,}\; \frac{\Bigl[\,1-\bigl(1-qp-(1-p)\,v\bigr)^{n}\Bigr]}{(1-p)(1-(1-v)^n)}, \inf_{v \in [q,1]} \frac{(1-p)\,\bigl[\,1-(1-q)^{n}\,\bigr]}{(1-p) (1-(1-v)^n)} \right\} \notag \\
    & \geq \min \left\{ \lim_{v \to 0} \frac{(v/q)}{\,p+(1-p)\,(v/q)\,}\;
    \frac{\Bigl[\,1-\bigl(1-qp-(1-p)\,v\bigr)^{n}\Bigr]}{(1-(1-v)^n)}, \lim_{v \to 1} \frac{1-(1-q)^{n}}{1-(1-v)^n} \right\}\notag \\
    & = \min \left\{ \frac{1-(1-qp)^n}{npq}, 1-(1-q)^n \right\} \geq \min\left\{ \frac{1-e^{-p\lambda}}{p\lambda}, 1-e^{-\lambda} \right \}. \label{exp:cr_largest}
    % \geq \min \left\{ \frac{1-(1-q)^n}{nq}, 1-(1-q)^n \right\} \\
    % & = \min \left\{ \frac{1-(1-\lambda/n)^n}{\lambda}, 1-(1-\lambda/n)^n \right\} = \min \left\{ \frac{1-e^{-\lambda}}{\lambda}, 1-e^{-\lambda} \right\} = 1-e^{-1}
    \end{align}
    Here, the second inequality uses the monotonicity of both functions (we prove that the first term is non-decreasing in the Lemma~\ref{lem:func_non_decreasing}), and the final inequality follows by substituting $q = \lambda/n$. 
    }

\section{Final Remarks}\label{sec:final_rem}

We introduced the Online Selection with Uncertain Disruption (\ppp) problem, which captures unexpected disruptions resulting from serving requests.
We first provided a non-adaptive single-threshold algorithm with a tight competitive ratio of $1 - 1/e$. We then analyzed the general class of adaptive threshold algorithms and showed that an asymptotic competitive ratio of $\theta^* \approx 0.745$ is attainable, and this is tight.

{ Even though in this work we focus on the case of fixed disruption probability $p$, we can use the techniques developed for the non-adaptive single-threshold algorithms to analyze a rare disruption regime, in particular, the case when $p=\alpha/n$ with $\alpha \leq 1$. Indeed, letting $q=1$ in Lemma \ref{lem:general_formula} and using the fact that $1-x\leq e^{-x}$ for all $x \in \mathbb{R}$, we obtain the following lower bound on the competitive ratio: $\left( \frac{1-e^{-\alpha}}{\alpha}\right)\min\left\{1, \frac{\alpha}{1-(1-\alpha/n)^n} \right\}=(1-e^{-\alpha})/\alpha$.
This competitive ratio is larger than $ 1-1/e$ for $\alpha\in [0,1)$, improving the ratio from the non-adaptive case for fixed $p$. Moreover, this asymptotic competitive ratio converges to $1$ as $\alpha \to 0$. }

It remains open the question of determining a constant lower bound for $\gamma_{n, p, { \zeta}}^{\mathrm{AD}}$ for all $n\geq 1$. Using the linear programming approach in~\cite{perez2025iid} (see also~\cite{jiang2023tightness}), we can approximate $\gamma_{n,p, { \zeta}}^{\mathrm{AD}}$ numerically (see Figure~\ref{fig:cr_final}). We empirically observe that as $n$ grows, $\gamma_{n,p, { \zeta}}^{\mathrm{AD}}$ converges to $\theta^*$. We also note that, for finite $n$, when the disruption probability is close to $1$, the problem aligns closely with the classical single-selection i.i.d.\ prophet inequality problem. 
Indeed, although \ppp\ allows multiple selections, if $p=1-\varepsilon$ with $\varepsilon\approx 0$, the expected gain beyond the first selection is multiplied by $O(\varepsilon^2)$. Therefore, this regime yields limited new insight, as it reduces to the single-selection case.

\begin{figure}[H]
    \centering
    \includegraphics[width=0.5\linewidth]{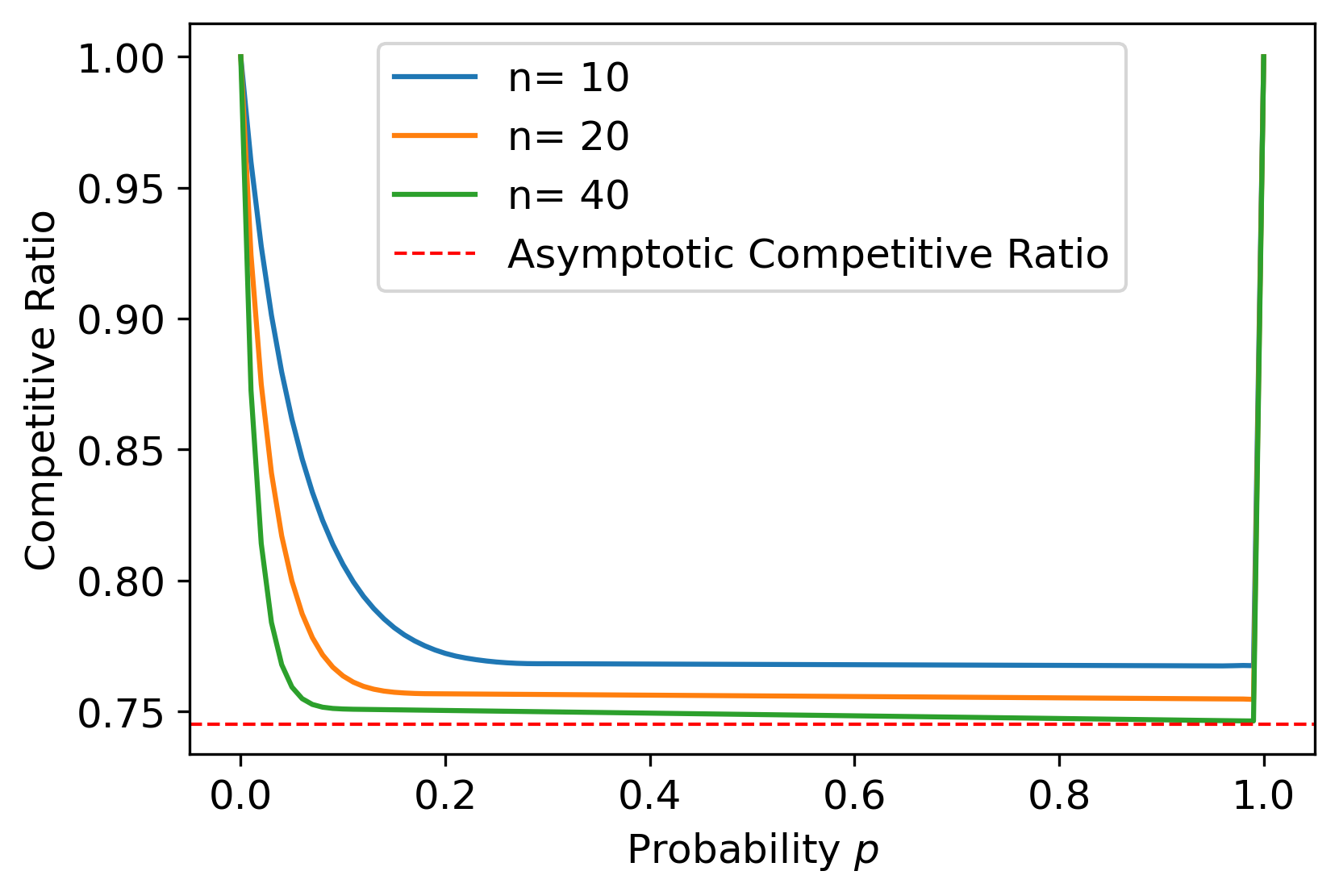}
    \caption{Estimated competitive ratios using linear programming formulation in \cite{perez2025iid} for $p\in [0,1]$.}
    \label{fig:cr_final}
\end{figure}

{  We have so far focused on the i.i.d.\ case of the \ppp{} problem, where all values share the same distribution. It is natural to consider the case, where the values remain independent, but are not necessarily identically distributed. In this case, using an approach akin to contention resolution schemes (see, e.g., ~\citep{alaei2014bayesian}), we can prove a competitive ratio of at least $1/2$, and this bound is tight (see \S\ref{app:sec4}). }

In this paper, we modeled the disruption as a memoryless process. A natural extension would be to consider broader classes of disruption processes. For example, the disruption probability might increase as more values are being accepted, modeling a system that wears out over time; conversely, it could also decrease, representing a system that becomes more reliable over time. However, analyzing such variants is non-trivial, as optimal algorithms must account for the number of selections made—a known challenge in multiple-selection problems and an active area of research (e.g., see ~\cite{alaei2014bayesian,jiang2023tightness,brustle2024splitting}). %In Example~\ref{example:vary_p}, we show that the competitive ratio in the \ppp problem can be zero for general disruption processes. \sk{I think we already mention the last sentence in the beginning, but there is no harm of repetition}

%Furthermore, the competitive ratio can depend on this disruption process.

%\sk{start up thing is unnecessary I think, we can just say becomes more stable reliable etc.}

{  Our paper sheds light on quantifying the \textit{price} of having limited information, and we attempt to study the value of \textit{flexibility} by showing how much adaptivity improves the competitive ratio over non-adaptive algorithms. An alternative way to measure the value of flexibility would be to benchmark non-adaptive algorithms directly against the optimal online policy. While conceptually compelling, this is analytically delicate; it remains an open question to establish such tight and distribution-free comparisons by exploiting the structure of the optimal online policy characterized by a dynamic program.    }

%\sk{Yihua, check above, I rewrote your paragraph, if it reads good drop below.} \yx{Yes, looks good! }

% {  Benchmarking non-adaptive policies directly against the optimal online policy is conceptually compelling but analytically delicate; a rigorous treatment is left for future work. Our focus is the \emph{price of not knowing}—quantifying loss from limited information. We also attempt to quantify the value of \emph{flexibility}, by showing that allowing adaptivity improves the competitive ratio over non-adaptive thresholds. }

{  
Lastly, \ppp assumes a single disruption event that terminates service with partial recovery on the last value. Another interesting direction is to consider models, in which a penalty is introduced on the last acceptance if a disruption occurs---for example, a fixed charge \(c>0\) or a linear penalty \(\phi\,X_{\mathrm{last}}\), akin to those studied in knapsack settings (e.g., see \cite{dean2008approximating,fu2018ptas}). In such cases the realized payoff may become negative, and the standard competitive benchmark can be ill-posed. To obtain meaningful guarantees, one should adopt alternative performance criteria, e.g. regret-type guarantees such as minimizing \(\E[v(\OPT)-v(\ALG)]\). These formulations are natural in revenue-management applications (e.g., airline overbooking or ticketing with compensation for denials), where the objective is expected revenue minus penalties subject to service constraints. 
}

% Lastly, \ppp assumes a single disruption event that terminates service without collecting the last value. Another interesting direction is to incorporate penalties and partial recoveries on the last accepted value, akin to those studied in knapsack settings~\citep{dean2008approximating,fu2018ptas}. This is particularly relevant in practical applications such as ticketing and reservations. 

\bibliographystyle{abbrvnat}
\bibliography{Revision}

\begin{thebibliography}{46}
\providecommand{\natexlab}[1]{#1}
\providecommand{\url}[1]{\texttt{#1}}
\expandafter\ifx\csname urlstyle\endcsname\relax
  \providecommand{\doi}[1]{doi: #1}\else
  \providecommand{\doi}{doi: \begingroup \urlstyle{rm}\Url}\fi

\bibitem[Alaei(2014)]{alaei2014bayesian}
S.~Alaei.
\newblock Bayesian combinatorial auctions: Expanding single buyer mechanisms to many buyers.
\newblock \emph{SIAM Journal on Computing}, 43\penalty0 (2):\penalty0 930--972, 2014.

\bibitem[Alijani et~al.(2020)Alijani, Banerjee, Gollapudi, Munagala, and Wang]{alijani2020predict}
R.~Alijani, S.~Banerjee, S.~Gollapudi, K.~Munagala, and K.~Wang.
\newblock Predict and match: Prophet inequalities with uncertain supply.
\newblock \emph{Proceedings of the ACM on Measurement and Analysis of Computing Systems}, 4\penalty0 (1):\penalty0 1--23, 2020.

\bibitem[Allouah et~al.(2023)Allouah, Bahamou, and Besbes]{allouah2023optimal}
A.~Allouah, A.~Bahamou, and O.~Besbes.
\newblock Optimal pricing with a single point.
\newblock \emph{Management Science}, 69\penalty0 (10):\penalty0 5866--5882, 2023.

\bibitem[Aouad and Sarita{\c{c}}(2020)]{aouad2020dynamic}
A.~Aouad and {\"O}.~Sarita{\c{c}}.
\newblock Dynamic stochastic matching under limited time.
\newblock In \emph{Proceedings of the 21st ACM Conference on Economics and Computation}, pages 789--790, 2020.

\bibitem[Arnosti and Ma(2023)]{arnosti2023tight}
N.~Arnosti and W.~Ma.
\newblock Tight guarantees for static threshold policies in the prophet secretary problem.
\newblock \emph{Operations research}, 71\penalty0 (5):\penalty0 1777--1788, 2023.

\bibitem[Assaf and Samuel-Cahn(2000)]{assaf2000simple}
D.~Assaf and E.~Samuel-Cahn.
\newblock Simple ratio prophet inequalities for a mortal with multiple choices.
\newblock \emph{Journal of applied probability}, 37\penalty0 (4):\penalty0 1084--1091, 2000.

\bibitem[Brustle et~al.(2025)Brustle, Perez-Salazar, and Verdugo]{brustle2024splitting}
J.~Brustle, S.~Perez-Salazar, and V.~Verdugo.
\newblock Splitting guarantees for prophet inequalities via nonlinear systems.
\newblock \emph{Mathematics of Operations Research}, 2025.

\bibitem[Chawla et~al.(2007)Chawla, Hartline, and Kleinberg]{chawla2007algorithmic}
S.~Chawla, J.~D. Hartline, and R.~Kleinberg.
\newblock Algorithmic pricing via virtual valuations.
\newblock In \emph{Proceedings of the 8th ACM Conference on Electronic Commerce}, pages 243--251, 2007.

\bibitem[Chawla et~al.(2010)Chawla, Hartline, Malec, and Sivan]{chawla2010multi}
S.~Chawla, J.~D. Hartline, D.~L. Malec, and B.~Sivan.
\newblock Multi-parameter mechanism design and sequential posted pricing.
\newblock In \emph{Proceedings of the forty-second ACM symposium on Theory of computing}, pages 311--320, 2010.

\bibitem[Chawla et~al.(2024)Chawla, Devanur, and Lykouris]{chawla2024static}
S.~Chawla, N.~Devanur, and T.~Lykouris.
\newblock Static pricing for multi-unit prophet inequalities.
\newblock \emph{Operations Research}, 72\penalty0 (4):\penalty0 1388--1399, 2024.

\bibitem[Cohen et~al.(2019)Cohen, Keller, Mirrokni, and Zadimoghaddam]{cohen2019overcommitment}
M.~C. Cohen, P.~W. Keller, V.~Mirrokni, and M.~Zadimoghaddam.
\newblock Overcommitment in cloud services: Bin packing with chance constraints.
\newblock \emph{Management Science}, 65\penalty0 (7):\penalty0 3255--3271, 2019.

\bibitem[Correa et~al.(2019)Correa, Foncea, Pizarro, and Verdugo]{correa2019pricing}
J.~Correa, P.~Foncea, D.~Pizarro, and V.~Verdugo.
\newblock From pricing to prophets, and back!
\newblock \emph{Operations Research Letters}, 47\penalty0 (1):\penalty0 25--29, 2019.

\bibitem[Correa et~al.(2021)Correa, Foncea, Hoeksma, Oosterwijk, and Vredeveld]{correa2021posted}
J.~Correa, P.~Foncea, R.~Hoeksma, T.~Oosterwijk, and T.~Vredeveld.
\newblock Posted price mechanisms and optimal threshold strategies for random arrivals.
\newblock \emph{Mathematics of operations research}, 46\penalty0 (4):\penalty0 1452--1478, 2021.

\bibitem[David and Nagaraja(2004)]{david2004order}
H.~A. David and H.~N. Nagaraja.
\newblock \emph{Order statistics}.
\newblock John Wiley \& Sons, 2004.

\bibitem[Dean et~al.(2008)Dean, Goemans, and Vondr{\'a}k]{dean2008approximating}
B.~C. Dean, M.~X. Goemans, and J.~Vondr{\'a}k.
\newblock Approximating the stochastic knapsack problem: The benefit of adaptivity.
\newblock \emph{Mathematics of Operations Research}, 33\penalty0 (4):\penalty0 945--964, 2008.

\bibitem[Delong et~al.(2024)Delong, Farhadi, Niazadeh, Sivan, and Udwani]{delong2022online}
S.~Delong, A.~Farhadi, R.~Niazadeh, B.~Sivan, and R.~Udwani.
\newblock Online bipartite matching with reusable resources.
\newblock \emph{Mathematics of Operations Research}, 49\penalty0 (3):\penalty0 1825--1854, 2024.

\bibitem[Epstein and Ma(2024)]{epstein2024selection}
B.~Epstein and W.~Ma.
\newblock Selection and ordering policies for hiring pipelines via linear programming.
\newblock \emph{Operations Research}, 72\penalty0 (5):\penalty0 2000--2013, 2024.

\bibitem[Feng et~al.(2025)Feng, Li, Li, Wu, and Wu]{feng2025iid}
Y.~Feng, B.~Li, H.~Li, X.~Wu, and Y.~Wu.
\newblock Iid prophet inequality with a single data point.
\newblock \emph{Artificial Intelligence}, 341:\penalty0 104296, 2025.

\bibitem[Fu et~al.(2018)Fu, Li, and Xu]{fu2018ptas}
H.~Fu, J.~Li, and P.~Xu.
\newblock {A PTAS for a Class of Stochastic Dynamic Programs}.
\newblock In \emph{45th International Colloquium on Automata, Languages, and Programming (ICALP 2018)}, volume 107 of \emph{Leibniz International Proceedings in Informatics (LIPIcs)}, pages 56:1--56:14, Dagstuhl, Germany, 2018. Schloss Dagstuhl -- Leibniz-Zentrum f{\"u}r Informatik.
\newblock ISBN 978-3-95977-076-7.

\bibitem[Goyal and Ravi(2010)]{goyal2010ptas}
V.~Goyal and R.~Ravi.
\newblock A ptas for the chance-constrained knapsack problem with random item sizes.
\newblock \emph{Operations Research Letters}, 38\penalty0 (3):\penalty0 161--164, 2010.

\bibitem[Gupta(2024)]{gupta2024greedy}
V.~Gupta.
\newblock Greedy algorithm for multiway matching with bounded regret.
\newblock \emph{Operations Research}, 72\penalty0 (3):\penalty0 1139--1155, 2024.

\bibitem[Hajiaghayi et~al.(2007)Hajiaghayi, Kleinberg, and Sandholm]{hajiaghayi2007automated}
M.~T. Hajiaghayi, R.~Kleinberg, and T.~Sandholm.
\newblock Automated online mechanism design and prophet inequalities.
\newblock In \emph{Proceedings of the 22nd National Conference on Artificial Intelligence - Volume 1}, AAAI'07, page 58–65, Vancouver, British Columbia, Canada, 2007. AAAI Press.
\newblock ISBN 9781577353232.

\bibitem[Harb(2025)]{harb2025new}
E.~Harb.
\newblock New prophet inequalities via poissonization and sharding.
\newblock In \emph{Proceedings of the 2025 Annual ACM-SIAM Symposium on Discrete Algorithms (SODA)}, pages 1222--1269. SIAM, 2025.

\bibitem[He et~al.(2025)He, Wei, Xu, and Yu]{he2025online}
S.~He, Y.~Wei, J.~Xu, and S.~H. Yu.
\newblock Online resource allocation without re-solving: The effectiveness of primal-dual policies.
\newblock \emph{Available at SSRN}, 2025.

\bibitem[Hill and Kertz(1982)]{hill1982comparisons}
T.~P. Hill and R.~P. Kertz.
\newblock {Comparisons of Stop Rule and Supremum Expectations of I.I.D. Random Variables}.
\newblock \emph{The Annals of Probability}, 10\penalty0 (2):\penalty0 336 -- 345, 1982.

\bibitem[Hill and Kertz(1992)]{hill1992survey}
T.~P. Hill and R.~P. Kertz.
\newblock A survey of prophet inequalities in optimal stopping theory.
\newblock \emph{Contemp. Math}, 125:\penalty0 191--207, 1992.

\bibitem[Hill and Krengel(1991)]{hill1991minimax}
T.~P. Hill and U.~Krengel.
\newblock Minimax-optimal stop rules and distributions in secretary problems.
\newblock \emph{The Annals of Probability}, pages 342--353, 1991.

\bibitem[Jiang et~al.(2023)Jiang, Ma, and Zhang]{jiang2023tightness}
J.~Jiang, W.~Ma, and J.~Zhang.
\newblock Tightness without counterexamples: A new approach and new results for prophet inequalities.
\newblock In \emph{Proceedings of the 24th ACM Conference on Economics and Computation}, EC '23, page 909, New York, NY, USA, 2023. Association for Computing Machinery.
\newblock ISBN 9798400701047.

\bibitem[Kerimov et~al.(2024)Kerimov, Ashlagi, and Gurvich]{kerimov2024dynamic}
S.~Kerimov, I.~Ashlagi, and I.~Gurvich.
\newblock Dynamic matching: Characterizing and achieving constant regret.
\newblock \emph{Management Science}, 70\penalty0 (5):\penalty0 2799--2822, 2024.

\bibitem[Kerimov et~al.(2025)Kerimov, Ashlagi, and Gurvich]{kerimov2025optimality}
S.~Kerimov, I.~Ashlagi, and I.~Gurvich.
\newblock On the optimality of greedy policies in dynamic matching.
\newblock \emph{Operations Research}, 73\penalty0 (1):\penalty0 560--582, 2025.

\bibitem[Kertz(1986)]{kertz1986stop}
R.~P. Kertz.
\newblock Stop rule and supremum expectations of iid random variables: a complete comparison by conjugate duality.
\newblock \emph{Journal of multivariate analysis}, 19\penalty0 (1):\penalty0 88--112, 1986.

\bibitem[Krengel and Sucheston(1977)]{krengel1977semiamarts}
U.~Krengel and L.~Sucheston.
\newblock {Semiamarts and finite values}.
\newblock \emph{Bulletin of the American Mathematical Society}, 83\penalty0 (4):\penalty0 745 -- 747, 1977.

\bibitem[Liu et~al.(2020)Liu, Leme, Pal, Schneider, and Sivan]{liu2020variable}
A.~Liu, R.~P. Leme, M.~Pal, J.~Schneider, and B.~Sivan.
\newblock Variable decomposition for prophet inequalities and optimal ordering, 2020.

\bibitem[Ma(2018)]{ma2018improvements}
W.~Ma.
\newblock Improvements and generalizations of stochastic knapsack and markovian bandits approximation algorithms.
\newblock \emph{Mathematics of Operations Research}, 43\penalty0 (3):\penalty0 789--812, 2018.

\bibitem[Mehta and Panigrahi(2012)]{mehta2012online}
A.~Mehta and D.~Panigrahi.
\newblock Online matching with stochastic rewards.
\newblock In \emph{Proceedings of the 2012 IEEE 53rd Annual Symposium on Foundations of Computer Science}, FOCS '12, page 728–737, USA, 2012. IEEE Computer Society.
\newblock ISBN 9780769548746.

\bibitem[Naor(1969)]{naor1969regulation}
P.~Naor.
\newblock The regulation of queue size by levying tolls.
\newblock \emph{Econometrica: journal of the Econometric Society}, pages 15--24, 1969.

\bibitem[Perez-Salazar and Verdugo(2024)]{perezsalazar2024optimalguaranteesonlineselection}
S.~Perez-Salazar and V.~Verdugo.
\newblock Optimal guarantees for online selection over time, 2024.

\bibitem[Perez-Salazar et~al.(2022)Perez-Salazar, Menache, Singh, and Toriello]{perez2022dynamic}
S.~Perez-Salazar, I.~Menache, M.~Singh, and A.~Toriello.
\newblock Dynamic resource allocation in the cloud with near-optimal efficiency.
\newblock \emph{Operations Research}, 70\penalty0 (4):\penalty0 2517--2537, 2022.

\bibitem[Perez-Salazar et~al.(2025)Perez-Salazar, Singh, and Toriello]{perez2025iid}
S.~Perez-Salazar, M.~Singh, and A.~Toriello.
\newblock The iid prophet inequality with limited flexibility.
\newblock \emph{Mathematics of Operations Research}, 2025.

\bibitem[Samuel-Cahn(1984)]{samuel1984comparison}
E.~Samuel-Cahn.
\newblock Comparison of threshold stop rules and maximum for independent nonnegative random variables.
\newblock \emph{The Annals of Probability}, 12\penalty0 (4):\penalty0 1213--1216, 1984.
\newblock ISSN 00911798, 2168894X.

\bibitem[Samuel-Cahn(1996)]{samuel1996optimal}
E.~Samuel-Cahn.
\newblock Optimal stopping with random horizon with application to the full-information best-choice problem with random freeze.
\newblock \emph{Journal of the American Statistical Association}, 91\penalty0 (433):\penalty0 357--364, 1996.

\bibitem[Shiryaev(2007)]{shiryaev2007optimal}
A.~N. Shiryaev.
\newblock \emph{Optimal stopping rules}, volume~8.
\newblock Springer Science \& Business Media, Berlin, Germany, 2007.

\bibitem[Sion(1958)]{sion1958general}
M.~Sion.
\newblock On general minimax theorems.
\newblock \emph{Pacific Journal of Mathematics}, 1958.

\bibitem[Van~Mieghem(1995)]{van1995dynamic}
J.~A. Van~Mieghem.
\newblock Dynamic scheduling with convex delay costs: The generalized c| mu rule.
\newblock \emph{The Annals of Applied Probability}, pages 809--833, 1995.

\bibitem[Wei et~al.(2023)Wei, Xu, and Yu]{wei2023constant}
Y.~Wei, J.~Xu, and S.~H. Yu.
\newblock Constant regret primal-dual policy for multi-way dynamic matching.
\newblock In \emph{Abstract Proceedings of the 2023 ACM SIGMETRICS International Conference on Measurement and Modeling of Computer Systems}, pages 79--80, 2023.

\bibitem[Zhang and Jaillet(2023)]{zhang2023secretary}
J.~Zhang and P.~Jaillet.
\newblock Secretary problems with random number of candidates: How prior distributional information helps, 2023.

\end{thebibliography}

\appendix

\newpage

\section{Missing Proofs from Section~\ref{sec: prelim}} \label{app:sec0}

\begin{proof}[Proof of Proposition~\ref{prop:decreasing_thres}]
    First, note that $\tau_n = 0$, since accepting always brings a non-negative value. Assume that an optimal algorithm, denoted by \( \ALG^{\text{OPT}} \), has $\tau_i < \tau_{i+1}$ for some $i \in [n-1]$. Consider the following alternative algorithm \( \ALG^{\text{ALT}} \) that swaps these two thresholds, i.e., the alternative algorithm accepts $X_i$ with threshold $\tau_{i+1}$, accepts $X_{i+1}$ with threshold $\tau_i$, and rest of the thresholds remain unchanged. Denote by \( v(\ALG^{\text{OPT}}) \) and \( v(\ALG^{\text{ALT}}) \) the expected total values collected under these two algorithms. Denote the expected total value obtained by both algorithms before observing the $i^{\text{th}}$ value by $v_{[1, \dots, i-1]}$. Conditioned on observing the $(i+2)^{\text{th}}$ value, denote the expected total value obtained by both algorithms starting from observing $(i+2)^{\text{th}}$ value by $v_{[i+2, \dots, n]}$. Let $X^{\tau_i} := \E[X_i|X_i \geq \tau_i]$ and  $p_i := \Pr[X_i \geq \tau_i]$. Finally, let $p_r$ be the probability that the algorithm observes value $X_i$, and $c^i = v_{[1, \dots, i-1]} + p_r (1-p_{i+1} p) (1-p_i p) v_{[i+2, \dots, n]}$. Then the expected total values for both algorithms can be written as
    \begin{align*}
        v(\ALG^{\text{OPT}}) & = c^i + p_r X^{\tau_i} p_i (1-p+{ p\zeta}) + p_r (1-p_i p) X^{\tau_{i+1}} p_{i+1}(1-p+{ p\zeta}), \\
        v(\ALG^{\text{ALT}}) & = c^i + p_r X^{\tau_{i+1}} p_{i+1} (1-p + { p\zeta}) + p_r (1-p_{i+1} p) X^{\tau_{i}} p_{i}(1-p+{ p\zeta}),
    \end{align*}
    which yields $v(\ALG^{\text{ALT}}) - v(\ALG^{\text{OPT}}) = p_r (1-p + { p\zeta}) p p_i p_{i+1} (X_{\tau_{i+1}} - X_{\tau_i}) \geq 0.$
    Thus, one can construct an alternative algorithm that employs a non-increasing sequence of thresholds by iteratively swapping thresholds, which achieves optimality. 
\end{proof}

\iffalse
\begin{claim}
    \label{lem:ordered_stat}
    For any $j \in [n]$, $\E[X_{(j)}] = \int_{0}^{1} \Pr[\text{bin}(n,v) \geq j] \cdot r(v) \, \mathrm{d}v$. 
\end{claim}

\begin{proof}[Proof of Claim~\ref{lem:ordered_stat}]
For any $j \in [n]$, we can write the expected value of $j^{\text{th}}$ top ordered statistics as follows: 
\begin{align}
    \quad \E[X_{(j)}] & = \int_{0}^{\infty} xf(x)F(x)^{n-j}(1-F(x))^{j-1} \cdot j \cdot \binom{n}{j} \, \mathrm{d}x \notag \\
    &= \int_{0}^{1} F^{-1}(1-u) \cdot (1-u)^{n-j} u^{j-1} \cdot j \cdot \binom{n}{j} \, \mathrm{d}u \notag \\
    & = \int_{0}^{1} \int_{u}^1 r(v) \, \mathrm{d}v \cdot (1-u)^{n-j} u^{j-1} \cdot j \cdot \binom{n}{j} \, \mathrm{d}u \notag \\
    &= \int_{0}^{1} \int_{0}^v (1-u)^{n-j} u^{j-1} \cdot j \cdot \binom{n}{j} \, \mathrm{d}u r(v) \, \mathrm{d}v \notag \\
    & = \int_{0}^{1} \Pr[\text{bin}(n,v) \geq j] \cdot r(v) \, \mathrm{d}v 
\end{align}
    where the second equality substitutes $1 - F(x) = u$, and the third equality comes from our assumption of $F^{-1}$ is differentiable and strictly decreasing, thus we use a function $r(v) > 0$ such that we have $\int_{u}^1 r(v) \, \mathrm{d}v = F^{-1}(1-u)$. 
\end{proof}
\fi

\begin{proof} [Proof of Proposition~\ref{prop:opt_value_complete}] { The first equality follows immediately from the following characterization of $v(\OPT)$, together with the subsequent claim, whose proof is provided at the end of the current proof.
$$v(\OPT) =\sum_{i=1}^n p(1-p)^{i-1}\!\left(\sum_{j=1}^{i-1}\mathbb{E}[X_{(j)}]+\zeta\,\mathbb{E}[X_{(i)}]\right) \;+\;(1-p)^n\sum_{j=1}^n \mathbb{E}[X_{(j)}].$$}
\begin{claim}
    \label{lem:ordered_stat}
    For any $j \in [n]$, $\E[X_{(j)}] = \int_{0}^{1} \Pr[\text{Bin}(n,v) \geq j] \cdot r(v) \, \mathrm{d}v$, where $r(v) > 0$ satisfies $\int_{u}^1 r(v)\, \mathrm{d}v = F^{-1}(1-u)$.
\end{claim}

\noindent To prove the second equality, we note that
 {  \begin{align*}
    v(\OPT) 
    &= \sum_{i=2}^n p (1-p)^{i-1} \sum_{j=1}^{i-1} \int_{0}^{1} F^{-1}(1-q)\, (1-q)^{n-j} q^{j-1}\, j \binom{n}{j}\, \mathrm{d}q \\ 
    &\hspace{5cm} + (1-p)^n  \sum_{j=1}^{n} \int_{0}^{1} F^{-1}(1-q)\, (1-q)^{n-j} q^{j-1}\, j \binom{n}{j}\, \mathrm{d}q \\[6pt]
    &\hspace{4.1cm} +\; \zeta \sum_{i=1}^n p (1-p)^{i-1} 
    \int_{0}^{1} F^{-1}(1-q)\, (1-q)^{n-1}\, q^{\,i-1}\, i \binom{n}{i}\, \mathrm{d}q \\[6pt]
    &= \int_{0}^{1} F^{-1}(1-q)\, 
    \sum_{j=1}^{n} (1-q)^{n-j} q^{j-1}\, j \binom{n}{j}\, \big((1-p)^j + \zeta p(1-p)^{j-1}\big)\, \mathrm{d}q \\[6pt]
    &= \int_{0}^{1} F^{-1}(1-q)\; g_n^{(\zeta)}(p,q)\, \mathrm{d}q, 
\end{align*}}
% \begin{align*}
%     v(\OPT) &= \sum_{i=2}^n p\cdot (1-p)^{i-1} \sum_{j=1}^{i-1} \int_{0}^{1} F^{-1}(1-q) \cdot (1-q)^{n-j} q^{j-1} \cdot j \cdot \binom{n}{j} \, \mathrm{d}q \\ 
%     & \hspace{5cm} + (1-p)^n  \sum_{j=1}^{n} \int_{0}^{1} F^{-1}(1-q) \cdot (1-q)^{n-j} q^{j-1} \cdot j \cdot \binom{n}{j} \, \mathrm{d}q\\
%     & = \int_{0}^{1} F^{-1}(1-q) \sum_{j=1}^{n-1} (1-q)^{n-j} q^{j-1} \cdot j \cdot \binom{n}{j} \sum_{i=j+1}^n p\cdot (1-p)^{i-1} \, \mathrm{d}q  \\ 
%     & \hspace{5.3cm} + \int_{0}^{1} F^{-1}(1-q) \sum_{j=1}^{n} (1-q)^{n-j} q^{j-1} \cdot j \cdot \binom{n}{j} (1-p)^n   \, \mathrm{d}q\\
%     & = \int_{0}^{1} F^{-1}(1-q) \sum_{j=1}^{n} (1-q)^{n-j} q^{j-1} \cdot j \cdot \binom{n}{j} \cdot (1-p)^j  \, \mathrm{d}q = \int_{0}^1 F^{-1}(1-q) g_n(p, q) \, \mathrm{d}q, 
% \end{align*}
where the last equality follows from the binomial theorem. 

\noindent \emph{Proof of Claim~\ref{lem:ordered_stat}.} \quad 
For any $j \in [n]$, let $b_{n,j} := j \cdot \binom{n}{j}$. Then we can write the expected value of the $j^{\text{th}}$ top ordered statistics as follows: 
\begin{align}
    \E[X_{(j)}] & = \int_{0}^{\infty} xf(x)F(x)^{n-j}(1-F(x))^{j-1} \cdot b_{n,j} \, \mathrm{d}x = \int_{0}^{1} F^{-1}(1-u) \cdot (1-u)^{n-j} u^{j-1} \cdot b_{n,j} \, \mathrm{d}u \notag \\
    & = \int_{0}^{1} \int_{u}^1 r(v) \, \mathrm{d}v \cdot (1-u)^{n-j} u^{j-1} \cdot b_{n,j} \, \mathrm{d}u = \int_{0}^{1} \Pr[\text{Bin}(n,v) \geq j] \cdot r(v) \, \, \mathrm{d}v, 
\end{align}
    where the second equality uses the substitution $1 - F(x) = u$, and the third equality comes from our assumption that $F^{-1}$ is differentiable and strictly decreasing so that the existence of $r(v)$ is guaranteed, and the last equality changes the order of integration. 
\end{proof}

\begin{example}\label{example:vary_p} 
Here, we show that varying the disruption probability within the selection process can lead to a competitive ratio of $0$. 
Let $n$ and $s$ be large constants, where $s \ll n$. Assume $n$ is divisible by $s$ for simplicity. Assume $X_i\sim \mathrm{Exp}(1)$.
Consider the disruption process $\Pr[Y_i=1]=1$ if $i = 1, 1+s, \ldots, 1 + (n/s-1)s$, and $0$ otherwise.
The optimal clairvoyant algorithm collects top $n/s$ ordered statistics in expectation, where the exponential distribution's order statistics roughly follows harmonic numbers (e.g., see \cite{david2004order}). Thus, $v(\OPT) \approx \sum_{l =1}^{n/s} \log\left( \frac{n}{l}\right) \approx  \left( \frac{n}{s} \log s + \frac{n}{s} - 1 \right).$
    On the other hand, the optimal online algorithm's value is $n/s$ by accepting every value when $Y_i = 0$. From here, we see that $v(\ALG)/v(\OPT) \rightarrow 0$ as $n$ and $s$ grow large. 
\end{example}

\section{Missing Proofs from Section~\ref{sec:static}} \label{app:sec1}

\begin{lemma}{\label{lemma_mono}} Fix $p \in (0,1)$. Then $\frac{v}{ \sum_{i=2}^n p\cdot (1-p)^{i-1} \sum_{j=1}^{i-1}\Pr[\text{Bin}(n,v) \geq j] + (1-p)^n \cdot nv}$ is a non-decreasing function in the interval $v \in [0, q]$.
\end{lemma}

\begin{proof}[Proof of Lemma \ref{lemma_mono}.]
    It is sufficient to show that the function 
    $f(v) = v/ \left( \sum_{j=1}^{i} \Pr[\text{Bin}(n,v) \geq j] \right)$ is non-decreasing for fixed  $i \in [n-1]$. Consider $g(v) = 1/f(v)$, and its derivative with respect to $v$: 
    \begin{align*}
        g'(v)&= \frac{1}{v^2} \left( \sum_{j=1}^{i} (\sum_{k = j}^{n}\binom{n}{k} kv^{k}(1-v)^{n-k} - \binom{n}{k} (n-k)v^{k+1}(1-v)^{n-k-1}) - \Pr[\text{Bin}(n,v) \geq j] \right)\\
        & = \frac{1}{v^2} \sum_{j=1}^{i} \frac{n!}{(n-j)!(j-1)!}v^j(1-v)^{n-j} - \frac{v^{n+1}}{1-v} - \Pr[\text{Bin}(n,v) \geq j]\\
        & = -\frac{i  v^{n+1}}{v^2 (1-v)} + \frac{1}{v^2} \sum_{j=1}^i (j \cdot \Pr[\text{Bin}(n,v) = j] - \Pr[\text{Bin}(n,v) \geq j]).
    \end{align*}
    Observe that $i v^{n+1}/ \left( v^2 (1-v) \right)$ is non-negative, and for each $k \in [1, i]$, $k \cdot \Pr[\text{Bin}(n,v) = k] - \sum_{j=1}^{k} \Pr[\text{Bin}(n,v) = j] = 0$, as $\Pr[\text{Bin}(n,v) \geq j] = \sum_{k=j}^{n} \Pr[\text{Bin}(n,v) = k]$. Thus, $ g'(v) \leq 0$, and $f(v)$ is non-decreasing in $v \in [0,q]$. 
\end{proof}

\begin{proof}[Proof of Lemma~\ref{lemma_hard_opt}]
    \begin{align*}
    v(\OPT) & \geq \sum_{i=2}^n p\cdot (1-p)^{i-1} \sum_{j=1}^{i-1} \int_{0}^{1} \Hat{F}(u) \cdot (1-u)^{n-j} u^{j-1} \cdot j \cdot \binom{n}{j} \, \mathrm{d}u \\
    & = \int_{0}^{1} \sum_{i=2}^n p\cdot (1-p)^{i-1} \sum_{j=1}^{i-1} \frac{a_1}{n}\delta_{\lbrace 0 \rbrace}(u) \cdot (1-u)^{n-j} u^{j-1} \cdot j \cdot \binom{n}{j} \, \mathrm{d}u \\
    & \hspace{3cm} + \int_{0}^{\beta/n} \sum_{i=2}^n p\cdot (1-p)^{i-1} \sum_{j=1}^{i-1} a_2 \cdot (1-u)^{n-j} u^{j-1} \cdot j \cdot \binom{n}{j} \, \mathrm{d}u\\
    & = a_1 \sum_{i=2}^n p(1-p)^{i-1} + a_2 \sum_{i=2}^n \sum_{j=1}^{i-1} p\cdot (1-p)^{i-1} \Pr[\text{Bin}(n, \beta/n) \geq j] \\
    & = a_1(1-p)(1-(1-p)^{n-1}) + a_2 \sum_{j=1}^{i-1} \Pr[\text{Bin}(n, \beta/n) \geq j] [(1-p)^j - (1-p)^n].
\end{align*} 
The second equality recognizes that $\Pr[\text{Bin}(n, \beta/n) \geq j] = \int_{0}^{\beta/n} (1-u)^{n-j} u^{j-1} \cdot j \cdot \binom{n}{j} \, \mathrm{d}u$, and we recover the lemma by letting $n \to \infty$. 
\end{proof}

\begin{proof}[Proof of Lemma~\ref{lemma_hard_threshold}]
We derive the performance of $\ALG^q$ via direct calculation.
\begin{align}
    & v(\ALG^q) = A_n(q, p) \cdot \frac{\int_0^q \Hat{F}(u) \, \mathrm{d}u}{q} = A_n(q, p) \cdot \frac{\int_0^q ({a_1}/{n})\delta_{\lbrace 0 \rbrace}(u) + a_2\mathbbm{I}_{(0, \beta/n]}(u)\, \mathrm{d}u}{q} \tag{per \eqref{exp:hard_dist}}\\
    & = \frac{A_n(q, p)}{q} \cdot \left(\frac{a_1}{n} + a_2 \cdot \min \{ q, \beta/n \} \right) \notag \\
    & = \max \left\lbrace \sup_{q \in [0, \beta/n]} \frac{A_n(q, p)}{q} \cdot \left(\frac{a_1}{n} + a_2 \cdot q \right),  \sup_{q \in [\beta/n, 1]} \frac{A_n(q, p)}{q} \cdot \left(\frac{a_1}{n} + \frac{a_2 \beta}{n} \right)\right\rbrace \notag \\
    & = \frac{1-p}{p} \cdot \max \left\lbrace \sup_{\lambda \in [0, \beta]} \frac{1}{\lambda} \left( 1 - \left(1-\frac{\lambda p}{n}\right)^n \right) \left( a_1 + a_2 \lambda \right), \sup_{\lambda \in [\beta, n]} \frac{1}{\lambda} \left( 1 - \left(1-\frac{\lambda p}{n}\right)^n \right) \left(a_1 + a_2 \beta \right) \right \rbrace,\label{exp_4.5_inter}
\end{align}
where the last equality follows from the change of variable $q= \lambda / n$. 
We denote by $r_1(\lambda)$ the second term in the maximum argument in~\eqref{exp_4.5_inter}. We have $r'_1(\lambda) = \frac{1}{\lambda^2} \left( \lambda p \left(1-\frac{\lambda p}{n}\right)^{n-1} + \left(1-\frac{\lambda p}{n}\right)^n - 1 \right)$, and for $\lambda = 2/p$, the derivative of $r_1$ becomes negative; hence,
 $\sup_{\lambda \in [\beta, n]} r_1(\lambda)$ can be restricted to $[\beta, \max \{\beta, 2/p\}]$. Next, letting $n \to \infty$ yields
\begin{align}\label{exp_4.5}
\lim_{n\to \infty} v(\ALG^q) = \frac{1-p}{p} \cdot \max \left\lbrace
\sup_{\lambda \in [0, \beta]} \frac{1-e^{-\lambda p}}{\lambda}  (a_1 + a_2 \lambda), \sup_{[\beta, \max \{\beta, 2/p\}]}\frac{1-e^{-\lambda p}}{\lambda} (a_1 + a_2 \beta)
\right\rbrace.
\end{align}

We denote by $r_2(\lambda)$ and $r_3(\lambda)$ the first and second terms in the maximum argument in~\eqref{exp_4.5}, respectively. Note that $r_2'(\lambda) = [a_1 (e^{-\lambda p} - 1) + p \lambda e^{-\lambda p} (a_1 + a_2 \lambda)]/\lambda^2$. Thus, $r_2{'}(\lambda^*) = 0$ implies that $e^{-\lambda^* p} \left( a_1 + a_1 p \lambda^* + a_2 p ({\lambda^{*}})^2 \right) = a_1 $; hence, $\sup_{\lambda \in [0, \beta]} r_2(\lambda)=\max\{r_2(\lambda^*),r_2(\beta),\lim_{\lambda\to 0} r_2(\lambda)\}$ where $0 < \lambda^* < \beta$, if it exists.
Moreover, $r_3(\lambda)$ is a decreasing function in $\lambda$, thus the supremum is attained at $\lambda = \beta$. We combine both cases to conclude the proof.
\end{proof}

\section{Missing Proofs from Section~\ref{sec:adapt}}\label{app:sec2}

\begin{proof} [Proof of Proposition~\ref{prop:ub_opt}]
By the second characterization of $v^\varepsilon(\OPT)$ in \eqref{exp:opt_value_complete}, we have
\begin{align*}
    v^\varepsilon(\OPT) & = \int_{0}^1 \Tilde{F}(u) g_n^{(0)}(p, u) \, \mathrm{d}u\\
    & = \int_{0}^1 \left( \frac{\theta^*}{(1-p)n} \cdot \delta_{\lbrace 0 \rbrace}(u) - \frac{p}{1-p} \int_{y^{-1} (e^{-pnu})}^{1-\varepsilon} \frac{1}{y'(s)} ds \mathbbm{I}_{(0, 1)}(u) \right) (1-p) \cdot n (1-pu)^{n-1} \, \mathrm{d}u \\
    & = \theta^* - \int_{0}^{1} \int_{y^{-1} (e^{-pnu})}^{1-\varepsilon} \frac{1}{y'(s)} ds (1-pu)^{n-1} pn \, \mathrm{d}u \\
    & = \theta^* - \int_{e^{-np}}^1 \int_{y^{-1} (v)}^{1-\varepsilon} \frac{1}{y'(s)} ds \left(1+\frac{\log(v)}{n}\right)^{n-1} \frac{\, \mathrm{d}v}{v} \tag{per $v = e^{-npu}$} \\
    & = \theta^* - \int_{0}^{1-\varepsilon} \int_{\max\{y(s), e^{-np}\}}^1 \left(1+\frac{\log(v)}{n}\right)^{n-1} \frac{\, \mathrm{d}v}{v} \frac{1}{y'(s)} ds \tag{changing order of integration}\\
    & = \theta^* - \int_{0}^{1-\varepsilon} \frac{1}{y'(s)} \left( 1- \left(1+\frac{\log(y(s))}{n}\right)^n\right) ds. \quad\text{(per $y(s) > e^{-np}$ when $s \in (0, 1-\varepsilon)$)} \qedhere
\end{align*}
\end{proof}

\begin{proof} [Proof of Lemma~\ref{lem: ub_adapt_be}]
    We first show that \eqref{exp: ub_adapt_sol} satisfies the Bellman equation \eqref{exp: ub_adapt_be}. 
    \begin{align*}
        & \int_{0}^{\mu} h(u) \, \mathrm{d}u -p \mu d(x) \\
        & = \int_{0}^{-\log(y(x))/p} \left( \theta^* \cdot \delta_{\lbrace 0 \rbrace}(u) - p \int_{y^{-1} (e^{-pu})}^{1} \frac{1}{y'(s)} ds \mathbbm{I}_{(0, 1)}(u) \right) \, \mathrm{d}u - p \cdot \frac{\log(y(x))}{p} \cdot \int_{x}^{1} \frac{1}{y'(s)} ds\\
        & = \theta^* - p \int_{0}^{1} \int_{0}^{\min\{-\log(y(x))/p, -\log(y(s))/p \}} \, \mathrm{d}u \frac{1}{y'(s)} ds - \log(y(x)) \int_{x}^{1} \frac{1}{y'(s)} ds\\
        & = \theta^* + \int_{0}^{1}  \frac{\max\{\log(y(x)), \log(y(s))\}}{y'(s)} ds - \int_{x}^{1} \frac{\log(y(x))}{y'(s)} ds \\
        & = \theta^* + \int_{0}^{x}  \frac{\log(y(s))}{y'(s)} ds + \int_{x}^{1}  \frac{\log(y(x))}{y'(s)} ds  -  \int_{x}^{1} \frac{\log(y(x))}{y'(s)} ds \tag{monotonicity of $y(x)$}\\
        & = \theta^* + \int_{0}^{x}  \frac{y''(s)}{(y'(s))^2} ds \tag{per $y''(x) = y'(x) \log(y(x))$}\\
        & = \theta^* - \frac{1}{y'(x)} + \frac{1}{y'(0)} = - \frac{1}{y'(x)} = -d'(x).
    \end{align*}

    Next, we show that $\mu$ is indeed the unique maximizer. We proceed in two steps: (i) the derivative with respect to $\mu$ is zero. (ii) the derivative is positive when $\mu = 0$, and the derivative is negative as $\mu \to \infty$. By Leibniz integral rule, we have $\frac{d}{d\mu} \left\lbrace \int_{0}^{\mu} h(u) \, \mathrm{d}u -p \mu d(x) \right\rbrace = h(\mu) - p \cdot d(x)$. 

    \noindent To verify (i), note that 
    \begin{align*}
        h(\mu) - p \cdot d(x) = \theta^* \cdot \delta_{\lbrace 0 \rbrace}(-\frac{\log(y(x))}{p}) - p \int_{x}^{1} \frac{1}{y'(s)} ds \mathbbm{I}_{(0, 1)}(\mu) - p \cdot \int_{x}^{1} -\frac{1}{y'(s)} ds = 0.
    \end{align*}

    \noindent To verify (ii), note that $\lim_{\mu \to \infty} h(\mu) - p \cdot d(x)  = -p \cdot d(x) < 0$ and
    \begin{equation*}
         h(0) - p \cdot d(x) = \theta^* - p \int_{0}^{1} \frac{1}{y'(s)} ds - p \cdot \int_{x}^{1} -\frac{1}{y'(s)} ds = \theta^* -p \cdot \int_{0}^{x} \frac{1}{y'(s)} ds > 0.\qedhere
    \end{equation*}

\end{proof}

\begin{proof} [Proof of Claim~\ref{clm: ub_tilde}]
    Let $d_{i, \sigma} := d((1-\sigma)i/n))$ and $y_{i,\sigma} := y((1-\sigma)i/n))$. Then, we have
    \begin{align*}
        & \hspace{0.3cm} (1-p) \int_{0}^q \Tilde{F}(u) \, \mathrm{d}u + (1-pq)(1+\eta_{\sigma})d_{i+1, \sigma} \\
        & \leq \frac{1}{n} \int_{0}^{\mu} h(u) \, \mathrm{d}u + \left(1-\frac{p\mu}{n}\right)(1+\eta_{\sigma}) \left( d_{i, \sigma} + \frac{1-\sigma}{n} d'_{i, \sigma} \right)\tag{Rewrite $q = \frac{\mu}{n}$ \& concavity of $d_{i, \sigma}$} \\
        & =  (1+\eta_{\sigma}) d_{i,\sigma} + \frac{1}{n} \left(\int_{0}^{\mu} h(u) \, \mathrm{d}u - (1+\eta_{\sigma})p\mu d_{i,\sigma} \right) + \frac{1}{n} (1+\eta_{\sigma}) \left( d'_{i, \sigma} (1-\sigma) - \frac{p\mu}{n} d'_{i, \sigma} (1-\sigma) \right)\\
        & \leq \widetilde{D}_i - \frac{d'_{i, \sigma}}{n} + \frac{\eta_\sigma \log(y_{i,\sigma})}{n} d_{i, \sigma} + \frac{d'_{i, \sigma}}{n} + \frac{\eta_\sigma d'_{i, \sigma}}{n} - \frac{\sigma d'_{i, \sigma}}{n} - \frac{\sigma \eta_\sigma d'_{i, \sigma}}{n} + \frac{\log(y_{i,\sigma})}{n^2 }d'_{i, \sigma} (1-\sigma) (1+\eta_\sigma)\\
        & = \widetilde{D}_i + \frac{\eta_\sigma \log(y_{i,\sigma})}{n} d_{i, \sigma} + \frac{\eta_\sigma d'_{i, \sigma}}{n} - \frac{\sigma d'_{i, \sigma}}{n} - \frac{\sigma \eta_\sigma d'_{i, \sigma}}{n} + \frac{\log(y_{i,\sigma})}{n^2 }d'_{i, \sigma} (1-\sigma) (1+\eta_\sigma)
    \end{align*}
    Note that $\eta_\sigma \log(y_{i,\sigma}) d_{i, \sigma} < 0$, and since $d'_{i, \sigma} < 0$, it is sufficient to show that $\eta_\sigma  - \sigma - \sigma \eta_\sigma + \frac{\log(y_{i,\sigma})}{n} (1-\sigma) (1+\eta_\sigma) \geq 0$, which is guaranteed by \eqref{exp: eta_req}. 
\end{proof}
%\sk{here}
{ 
\section{Missing Proofs from Section~\ref{sec:extensions}}\label{app:sec3}

%\sk{a bit weird to start with let $\mu_k=$ because the final expression appears in the statement of the lemma } \yx{how about now, i change the statement of the lemma. }

% \sk{what is little $f_\tau$} \yx{Rather than defining, let's change the formula we use in the first line :), directly integrating wrt $F_\tau(x)$}
\begin{proof}[Proof of Lemma~\ref{lem:single_max}]
    Recall that $\mu_l = \E[\max\{X_1,\dots,X_l\}| X_1\geq \tau,\dots,X_l\geq \tau]$. We first characterize $\mu_l$ via a positive function $r(v)$, where $\int_{u}^1 r(v)\mathrm{d}v = F^{-1} (1-u)$. Let $F_\tau(x) := \Pr[X_l\leq x | X_l \geq \tau] = (F(x) - F(\tau))/ (1- F(\tau))$. We have
    \begin{align*}
    \mu_l & = \int_{\tau}^{\infty} x\,l\,F_\tau(x)^{\,l-1} \mathrm{d} F_\tau(x) = \frac{1}{(1-F(\tau))^{l}}\int_{\tau}^{\infty} x\,l\,\bigl(F(x)-F(\tau)\bigr)^{l-1}\,\mathrm{d}F(x) \\
    & = \frac{1}{q^{l}}\int_{0}^{q} F^{-1}(1-w) l (q-w)^{l-1} \mathrm{d}w = \frac{l}{q^{l}}\int_{0}^{q} \int_{w}^1 r(v)\mathrm{d}v (q-w)^{l-1} \mathrm{d}w \\
    & = \frac{l}{q^{l}}\int_{0}^{1} \int_{0}^{\min \{q,v\}}  (q-w)^{l-1} \mathrm{d}w r(v)\mathrm{d}v = \frac{1}{q^{l}} \int_{0}^{1} \left( q^l-(q-\min \{q,v\})^l \right) r(v)\mathrm{d}v.
\end{align*}
We complete the proof by employing the same non-adaptive algorithm value formulation as established in Lemma~\ref{lem:general_formula}.
\end{proof}

\noindent Define $$R(v) := \frac{(v/q)}{\,p+(1-p)\,(v/q)\,}\; \frac{1-\bigl(1-qp-(1-p)\,v\bigr)^{n}}{(1-(1-v)^n)}.$$

\begin{lemma}\label{lem:func_non_decreasing}
    $R(v)$ is non-decreasing in $v \leq q$. 
\end{lemma}

\begin{proof}[Proof of Lemma \ref{lem:func_non_decreasing}.]
We begin by differentiating \(R(v)\) explicitly for \(n\le 4\) and verify, by direct computation, that \(R'(v)\ge 0\). We then turn to the case \(n\ge 5\). For $n \geq 5$, it is convenient to recast \(R(v)\) using the following shorthand notation:
$$
S(x):=\sum_{j=0}^{n-1}(1-x)^j=\frac{1-(1-x)^n}{x},\qquad
s(v):=qp+(1-p)\,v,\qquad
R(v)=\frac{S\!\big(s(v)\big)}{S(v)}.
$$
We take the derivative of $\log R$. Then we have
$$
\frac{d}{dv}\log R(v)
=\frac{d}{dv}\!\left[\log S\!\big(s(v)\big)-\log S(v)\right]
=(1-p)\,\frac{S'\!\big(s(v)\big)}{S\!\big(s(v)\big)}-\frac{S'(v)}{S(v)}.
$$
Next, we show that $\log S$ is convex for $n \geq 5$. Setting $u=1-x$  yields
$$
S'(x)=-\sum_{j=1}^{n-1} j\,u^{\,j-1},\qquad
S''(x)=\sum_{j=2}^{n-1} j(j-1)\,u^{\,j-2}.
$$
Hence,
$$
(\log S)''(x)
=\frac{S''(x)S(x)-\big(S'(x)\big)^2}{\big(S(x)\big)^2}
=\frac{\sum_{0\le i<j\le n-1}(i-j)^2\,u^{\,i+j-2} - \frac{1}{2}\sum_{i,j=0}^{n-1} (i+j) u^{i+j-2}}{\big(S(x)\big)^2}.
$$
When $n\geq 5$, $(\log S)''(x) \geq 0$, thus $(\log S)'(x)=S'(x)/S(x)$ is non-decreasing. Since $v\le q$, $s(v)=qp+(1-p)\,v\ \ge\ v.$
Because $(\log S)'$ is increasing and $0\le 1-p\le 1$,
$$
(1-p)\,(\log S)'\!\big(s(v)\big)-(\log S)'(v)
\ \ge\ (1-p)\,(\log S)'(v)-(\log S)'(v)
=-p\,(\log S)'(v)\ \ge\ 0,
$$
where the last inequality uses $(\log S)'(v)=S'(v)/S(v)\le 0$.
Therefore,
$$
\frac{d}{dv}\log R(v)=(1-p)\,\frac{S'\!\big(s(v)\big)}{S\!\big(s(v)\big)}-\frac{S'(v)}{S(v)}\ \ge\ 0,
$$
and hence, $R'(v)\ge 0$ on $(0,q]$.

% $$
% g(x)=1-(1-x)^n,\qquad
% h(x):=\frac{g'(x)}{g(x)}=\frac{n(1-x)^{n-1}}{1-(1-x)^n}>0.
% $$
% Let $c=1-p$, $t=v/q\in(0,1]$, and $s(v) =qp+cv$. Thus $R(v)$ can be rewritten as: 
% $$
% R(v)=\frac{t}{p+ct}\cdot\frac{g\big(s(v)\big)}{g(v)}.
% $$
% We calculate $\log R(v)$ and the derivative of $\log R(v)$ with respect to $v$. 
% $$
% \log R(v)=\log t-\log\!\big(p+c\,t\big)+\log g\big(s(v)\big)-\log g(v).
% $$
% $$
% \frac{d}{dv}\log R(v)
% =\left(\frac{1}{q\,t}-\frac{c}{q\,(p+c\,t)}\right)+c\,h\!\big(s(v)\big)-h(v)
% =\frac{p}{q\,t\,(p+c\,t)}+c\,h\!\big(s(v)\big)-h(v). \quad(\star)
% $$

% Now, we shift our focus to show that $(\star)$ is non-negative. 

% Since $v\le q$ we have $s(v)=qp+cv\ge v$. Because $g$ is concave, $h(x)=g'(x)/g(x)$ is decreasing, so
% $$
% h\big(s(v)\big)\le h(v)\ \Longrightarrow\ c\,h\big(s(v)\big)-h(v)\ \ge\ -p\,h(v).
% $$
% Plug into $(\star)$:
% $$
% \frac{d}{dv}\log R(v)\ \ge\ \frac{p}{q\,t\,(p+c\,t)}\;-\;p\,h(v)
% \ =\ p\!\left(\frac{1}{v\,(p+c\,t)}-h(v)\right).
% \quad (\star\star)
% $$

\end{proof}

\subsection{Upper Bound} \label{subsec:max_largest_ub}
In this subsection, we explain that our analysis is tight by illustrating an instance as the proof for the upper bound. We use the same distribution from~\eqref{exp:hard_dist}, that is, 
\begin{align}
    \Hat{F}(u) := \frac{a_1}{n} \delta_{\lbrace 0 \rbrace}(u) + a_2\mathbbm{I}_{(0, \beta/n]}(u). 
\end{align}
We can separately characterize the value of the optimal algorithm and the value of the single threshold algorithm, when $n \to \infty$. 
\begin{lemma}\label{lemma_max_largest_opt}
    When $n\to \infty$, $v_{\max}(\OPT) = (1-p) \cdot \Big[ a_1 + a_2\Big(1- e^{-\beta}\Big) \Big]$. 
\end{lemma}
\begin{proof}[Proof of Lemma \ref{lemma_max_largest_opt}] We use the expression from Lemma~\ref{lem:opt_max}, 
    \begin{align*} 
        v(\OPT) & = (1-p) \E[X_{(1)}] = (1-p) \int_{0}^{1} F^{-1}(1-u)\,(1-u)^{n-1} n \mathrm{d}u,\\
& = (1-p) n \int_{0}^{1}\frac{a_1}{n}\delta_{\{0\}}(u) (1-u)^{n-1}\,\mathrm{d}u\ +  (1-p) n\int_{0}^{\beta/n} a_2 (1-u)^{n-1} \mathrm{d}u. \\
& = (1-p) a_1 + (1-p) n \cdot a_2 \cdot \frac{1-(1-\beta/n)^n}{n} = (1-p) \cdot \Big[ a_1 + a_2\Big(1-(1-\beta/n)^n\Big) \Big]
    \end{align*}
    We recover the claim by letting $n\to \infty$. 
\end{proof}

\begin{lemma}\label{lemma_max_largest_alg}
As $n$ grows large, the expected value of $\ALG^{t/n}$ converges to
\[
\lim_{n \to \infty} v_{\max}(\ALG^{t/n}) = 
\sup_{t\in [0,\beta]} (1-p)\left[ a_2(1-e^{-t}) + \frac{a_1}{p}\,\frac{1-e^{-pt}}{t} \right]
\]
% Let $t=t^* > 0$ be a solution to the equation
% \begin{equation}
% a_2 e^{-t} +\frac{a_1}{p}\frac{(pt+1)e^{-pt}-1}{t^{2}} = 0
% \label{eq:deriva}
% \end{equation}
% As $n$ grows large, the expected value of $\ALG^{t/n}$ converges to
% \[
% \lim_{n \to \infty} v(\ALG^{t/n}) = \max \left\{ (1 - p)a_1,\; (1-p)\left[ a_2(1-e^{-\beta}) + \frac{a_1}{p}\,\frac{1-e^{-p\beta}}{\beta} \right],\; C_3 \right\},
% \]
% where $C_3=(1-p)\left[ a_2(1-e^{-t^*}) + \frac{a_1}{p}\,\frac{1-e^{-pt^*}}{t^*} \right]$ if $t^* \leq \beta$ and $C_3=0$ otherwise.
\end{lemma}

\begin{proof}[Proof of Lemma \ref{lemma_max_largest_alg}]
    Define $m = \min \{q, \beta/n\}$, we can find the following expressions for $\mu_k$ where $\mu_0 = 0$, and 
    \begin{align*}
    \mu_k & = \frac{1}{q^{k}}\int_{0}^{q} F^{-1}(1-w) k (q-w)^{k-1} \mathrm{d}w = \frac{1}{q^{k}}\!\left[\frac{a_1}{n}\,k\,q^{\,k-1} \;+\; a_2 \int_{0}^{m} k\,(q-w)^{k-1}\,\mathrm{d}w \right] \\
    &= \frac{1}{q^{k}}\!\left[\frac{a_1}{n}\,k\,q^{\,k-1} \;+\; a_2 \big(q^{k}-(q-m)^{k}\big) \right] = \frac{a_1}{n}\,\frac{k}{q} \;+\; a_2\!\left[1-\Big(1-\frac{m}{q}\Big)^{\!k}\right].
    \end{align*}
    Thus, using the expression from Lemma~\ref{lem:single_max} and denoting $b_{n,q}^a := \Pr[\mathrm{Bin}(n,q)=a]$, $c_k := 1-\Big(1-\frac{\beta}{nq}\Big)^{k}$, 
    \begin{align}
    &\sup_{q \in [0,1]} v_{\max}(\ALG^q)  = \max \Bigg\{ \sup_{q\in [0,\beta/n]} \sum_{a=1}^n b_{n,q}^a
    \left(\sum_{k=1}^{a-1} p(1-p)^k (\frac{a_1 k}{nq} +a_2) + (1-p)^a(\frac{a_1 a}{nq} +a_2)\right), \notag \\
    & \hspace{4.3cm} \sup_{q\in[\beta/n, 1]}  \sum_{a=1}^n b_{n,q}^a
    \left(\sum_{k=1}^{a-1} p(1-p)^k (\frac{a_1 k}{nq} +a_2 c_k) + (1-p)^a(\frac{a_1 a}{nq} +a_2 c_a)\right) \Bigg\} \notag \\
    & = \max \Bigg\{ \sup_{q\in [0,\beta/n]} a_2 (1-p)\big( 1-(1-q)^n\big)  + \frac{a_1}{n}\frac{1-p}{p}\frac{1-(1-qp)^{n}}{q},  \sup_{q\in[\beta/n, 1]} \frac{a_1}{n}\frac{1-p}{p}\frac{1-(1-qp)^{n}}{q} \notag \\
    & \hspace{2cm} + a_2\left[1 - \frac{p}{p+(1-p)\tfrac{\beta}{nq}}\Bigl(1-\bigl(1-qp-(1-p)\tfrac{\beta}{n}\bigr)^{n}\Bigr) - \bigl(1-qp-(1-p)\tfrac{\beta}{n}\bigr)^{n} \right] \Bigg\} \notag \\
    & = \max \Bigg\{ \sup_{t\in [0,\beta]}  (1-p)\left[ a_2\left(1-\Bigl(1-\frac{t}{n}\Bigr)^{n}\right) +\frac{a_1}{p}\,\frac{1-\left(1-\frac{pt}{n}\right)^{n}}{t} \right] , \sup_{t\in[\beta, n]}\frac{a_1(1-p)}{p}\frac{1-\left(1-\frac{pt}{n}\right)^{n}}{t}\notag \\ 
    & \hspace{1.5cm} + a_2\left[1 - \frac{p}{p+(1-p)\tfrac{\beta}{t}}\left(1-\left(1-\frac{pt+(1-p)\beta}{n}\right)^{n}\right) - \left(1-\frac{pt+(1-p)\beta}{n}\right)^{n}\right]\Bigg\},  \label{exp:sup_second_term}
    \end{align}
    where the last equality uses the change of variable $q = t/n$.  We denote by $r_1(t)$ the second term in the maximum argument in~\eqref{exp:sup_second_term}. Similar to the analysis in Lemma~\ref{lemma_hard_threshold}, we can prove that both terms in $r_1(t)$ is decreasing when $t> 2/p$ (see Claim~\ref{clm:decreasing_t} for details). Hence, $\sup_{t\in[\beta, n]} r_1(t)$ can be restricted to $[\beta, \max\{\beta , 2/p\}]$. Next, letting $n\to \infty$, we can extend~\eqref{exp:sup_second_term}: 
    \begin{gather}
        \sup_{q \in [0,1]} v_{\max}(\ALG^q) = \max \Bigg\{ \sup_{t\in [0,\beta]} (1-p)\left[ a_2(1-e^{-t}) + \frac{a_1}{p}\,\frac{1-e^{-pt}}{t} \right], \notag \\
        \sup_{t \in [\beta, \max\{\beta, 2/p\}]} \frac{a_1(1-p)}{p}\,\frac{1-e^{-pt}}{t} + a_2\!\left[ 1 - \frac{p}{\,p+(1-p)\tfrac{\beta}{t}}\bigl(1-e^{-(pt+(1-p)\beta)}\bigr) - e^{-(pt+(1-p)\beta)} \right]. \label{exp:sup_term_max}
    \end{gather}
    We denote by $r_2(t)$ the second term in the maximum argument in~\eqref{exp:sup_term_max}. $r_2(t)$ is a decreasing function in $t$, thus the supremum is attained at $t = \beta$. From here, we conclude the proof. 
\end{proof}
\begin{claim}\label{clm:decreasing_t}
$r_1(t)$ is decreasing in $t$ for $t>2/p$.
\end{claim}
\begin{proof}[Proof of Claim \ref{clm:decreasing_t}]
By Lemma~\ref{lemma_hard_threshold}, the derivative of the first term in $r_1(t)$ is negative for $t>2/p$. It remains to handle the second term. Set $y(t) =pt+(1-p)\beta$, 
so that the second term can be rewritten as 

$$ a_2\frac{(1-p)\beta}{y(t)}\Bigl(1-\Bigl(1-\frac{y(t)}{n}\Bigr)^{\!n}\Bigr).$$

\noindent Observe that the preceding expression has the same structure as the first term; hence its derivative is negative under the same condition, namely $y(t)> 2$. Since $t > 2/p$ implies $y(t) =pt+(1-p)\beta \geq pt > 2$, this condition is satisfied.
\end{proof}

Finally, we use the preceding results to establish that the competitive ratio admits an upper bound of \(1 - e^{-\lambda(p)}\). Specifically, we have
\begin{align*}
    & \inf_{F \in \mathcal{F}, n\geq 1} \sup_{q \in [0,1]} \frac{v_{\max}(\ALG^q)}{v_{\max}(\OPT)} \leq \inf_{n \geq 1} \min_{a_1,a_2\geq 0, a_1+a_2=1} \sup_{q \in [0,1]}  \frac{v_{\max}(\ALG^q)}{v_{\max}(\OPT)} \\
    & \leq \lim_{\beta \to \infty} \lim_{n \to \infty} \min_{a_1,a_2\geq 0,a_1+a_2=1} \sup_{t \geq 0} \frac{(1-p)\left[ a_2(1-e^{-t}) + \frac{a_1}{p}\,\frac{1-e^{-pt}}{t} \right]}{1-p} \\
    & = \max_{t\geq 0} \min_{a_1,a_2\geq 0,a_1+a_2=1}\left\{ a_2 (1-e^{-t})+ a_1 \left( \frac{1-e^{-pt}}{pt} \right)   \right\} = \sup_{t\geq 0} \min\left\{ 1-e^{-t} , \frac{1-e^{-pt}}{pt}.  \right\}
\end{align*}
The second inequality follows directly from Lemmas~\ref{lemma_max_largest_opt} and~\ref{lemma_max_largest_alg}.  
The first equality uses the quasi-concavity of the function $f(t)=a_2(1-e^{-t})+\frac{a_1}{p}\,\frac{1-e^{-pt}}{t}$ which allows the interchange of the min and max operators under Sion’s minimax theorem~\citep{sion1958general}. The proof is detailed in Claim~\ref{clm:quasiconcave}. The final equality then follows by noting that, for any fixed \(t\), the minimum over the convex combination parameters \(a_1,a_2\) occurs at one of the extreme points of the simplex. 
\begin{claim}\label{clm:quasiconcave}
    For any $a_1,a_2\geq 0,a_1+a_2=1, 0 <p <1,$ and $t>0$, $f(t)$ is a quasi-concave function. 
\end{claim}
\begin{proof}[Proof of Claim \ref{clm:quasiconcave}]
A direct differentiation gives
$$
f'(t)=a_2 e^{-t}+\frac{a_1}{p\,t^2}\Big(e^{-pt}(1+pt)-1\Big).
$$
Multiplying by the positive factor $p\,t^2$ and rearranging,
$$
f'(t)=0 \iff L(t)=R(t),
\quad
L(t):=a_1\Big(1-e^{-pt}(1+pt)\Big),\quad
R(t):=p\,a_2\,t^2 e^{-t}.
$$
$L$ is strictly increasing on $(0,\infty)$:
$$
L'(t)=a_1 p^2 t e^{-pt}>0,\qquad
\lim_{t\downarrow0}L(t)=0,\qquad
\lim_{t\to\infty}L(t)=a_1.
$$
\(R\) is unimodal with a unique maximum at \(t=2\):
$$
R'(t)=p a_2 e^{-t}(2t-t^2),\qquad
R(0)=0,\qquad
\lim_{t\to\infty}R(t)=0.
$$

Since $L$ is strictly increasing, and $R$ increases on $(0,2)$, then decreases on $(2,\infty)$,
the equation $L(t)=R(t)$ has at most one solution $t^*>0$.
Indeed, if \(0<t_1<t_2\) with $L(t_i)=R(t_i)$, then, because $L$ is strictly increasing,
$R-L>0$ on $(t_1,t_2)$. But $R-L$ has at most one local maximum (as $R$ has one
and $L$ is monotone), so starting from $R(0)-L(0)=0$, it cannot become positive and
return to $0$ twice without creating two distinct extrema.
Hence, $f'$ has at most one zero on $(0,\infty)$.
\end{proof}

% For any $a_1 = a_2$, to satisfy~$\eqref{eq:deriv}$ equal to zero, we require 
% {\color{red}We would like to find the relation between $a_1$ and $a_2$, such that~$\eqref{eq:deriv}$ is equal to zero, and ideally reach lower bound, which is $\min\left\{ \frac{1-e^{-p\lambda}}{p\lambda}, 1-e^{-\lambda} \right \}$. Thus, we set $\frac{1-e^{-pt}}{pt}=1-e^{-t}$. Differentiate both sides: 
% \[
% \frac{(pt+1)e^{-pt}-1}{t^{2}}=p\,e^{-t}.
% \]
% Substituting into~\eqref{eq:deriv}
% \[
% a_2 e^{-t}+\frac{a_1}{p}\,\frac{(pt+1)e^{-pt}-1}{t^{2}}=0
% \]
% gives
% $a_2 = -a_1$. WHY IS THIS HAPPENING AHHHHH}

% Thus, we find that the optimal value can only be found either at $t = 0$, or $t = \beta$ or $t$ equals the time when the derivative of following expression is equal to zero $ a_2 (1-e^{-t}) (1-p) + a_1\,\frac{1-p}{p} \frac{1-e^{-pt}}{t}$. We can take the following procedure to find a good instance, for any $p$, 
% \begin{enumerate}
%     \item Find $\lambda$ such that $\frac{1-e^{-p\lambda}}{p\lambda} = 1-e^{-\lambda}$.
%     \item Use $t = \lambda$ found above and $a_1 = 1$(can be any number), calculate $a_2 = -\,a_1\,e^{t}\!\left[\frac{e^{-pt}}{t} - \frac{1 - e^{-pt}}{p\,t^{2}}\right]$
%     \item Compute $v(\OPT) = (1-p) \cdot \Big[ a_1 + a_2\Big(1- e^{-\beta}\Big) \Big]$ for $\beta$ big enough and $v(\ALG) = a_2 (1-e^{-t}) (1-p) + a_1\,\frac{1-p}{p} \frac{1-e^{-pt}}{t}$
% \end{enumerate}

\section{Missing Details from Section~\ref{sec:final_rem}}\label{app:sec4}
In this section, we study the non-i.i.d.\ variant of \ppp, where the values $X_i$'s are drawn from non-negative, independent, but potentially non-identical  distributions \(F_i\)'s, respectively. For this variant, we first establish a lower bound for the competitive ratio against the benchmark, which can rearrange the arrival sequence after observing all $X_i$'s. We then conclude by providing a tight upper bound for the competitive ratio. 

\begin{theorem} \label{thm: non-iid}
    In the non-i.i.d.\ variant of \ppp with $\zeta\in[0,1]$, the optimal online algorithm achieves at least a $1/2$-competitive ratio with respect to the offline benchmark. 
    % when the disruption probabilities $p_i$ are allowed to vary.
\end{theorem}

\begin{proof}[Proof of Theorem \ref{thm: non-iid}]
For each $i\in[n]$, let $z_i$ denote the probability that $\OPT$ selects item $i$, and set $B = 1- (1-\zeta) p$. 
By stochastic dominance (conditioning on the top $z_i$-quantile maximizes the conditional mean), there exists a threshold $\tau_i$ with
$z_i=\Pr[X_i\ge\tau_i]$ such that
\begin{equation}\label{eq:OPT-upper}
v(\OPT)
=\sum_{i=1}^n B \cdot \E[X_i\mid \OPT\text{ selects }i] \cdot z_i
\leq \sum_{i=1}^n B \cdot \E[X_i\mid X_i\ge\tau_i] \cdot z_i.
\end{equation}
Consider the following online algorithm $\ALG$.
At time $i$, $\ALG$ independently skips the current item with probability $\varepsilon_i\in[0,1]$ (to be determined below); otherwise, if $X_i\ge\tau_i$, $\ALG$ accepts the current item. Let $R_i := \Pr[\ALG \text{ reaches item } i]$, where $R_1=1$. Then we can write the value of $\ALG$ as 
\begin{equation}\label{eq:ALG-value}
v(\ALG)=\sum_{i=1}^n R_i \cdot (1-\varepsilon_i) \cdot B \cdot \E[X_i\mid X_i\ge\tau_i] \cdot z_i.
\end{equation}

%We compare $v(\ALG)$ and $v(\OPT)$, and find that all we need to prove is that for all $i \in [n]$, $R_i (1-\epsilon_i)$ is at least $1/2$. 

\noindent Since both $v(\ALG)$ and $v(\OPT)$ are comparable term by term, it is sufficient to prove that for all $i \in [n]$, we have $R_i (1-\epsilon_i) \geq 1/2$. To see this, we first characterize a recurrence relation for $R_i$. From time $i$ to $i+1$, $\ALG$ fails to advance if and only if $\ALG$ does not skip (with probability ($1-\epsilon_i$)), $X_i\ge\tau_i$, and disruption happens.
Therefore, 
\begin{equation}\label{eq:R-rec}
R_{i+1} \;=\; R_i\bigl(1-(1-\varepsilon_i)\,z_i\,p\bigr),\qquad i=1,\dots,n-1.
\end{equation}
Moreover, to guarantee that $R_i (1-\epsilon_i) \geq 1/2$, we simply set $1-\epsilon_i = 1/(2R_i)$, and plugging into the recurrence \eqref{eq:R-rec}, we have the following updated recurrence for $R_i$: 
\begin{equation}\label{eq:R-linear}
R_{i+1} \;=\; R_i\!\left(1-\frac{z_i p}{2R_i}\right) \;=\; R_i - \frac{z_i p}{2}.
\end{equation}
Iterating~\eqref{eq:R-linear} yields the closed forms for $R_i$ and $\epsilon_i$ for all $i \in [n]$:
\begin{equation}\label{eq:closed-forms}
R_i = 1 - \frac{p}{2}\sum_{k=1}^{i-1} z_k,\qquad
% 1-\epsilon_i \;=\; \frac{1}{\,2 - \sum_{k=1}^{i-1} z_k p_k\,},\qquad
\varepsilon_i = 1 - \frac{1}{2 - p \sum_{k=1}^{i-1} z_k }.
\end{equation}
% To ensure every $R_i \geq 1/2$ and $\epsilon_i \in [0,1]$ a valid probability, we need to require $\sum_{k=1}^{i-1} z_k \leq 1/p$, which is gauranteed based on the definition of $z_k$, because even if $\OPT$ decides to select every item, $\sum_{i=1}^{n} z^k \leq \sum_{i=1}^n (1-p)^{i-1} = \frac{1-(1-p)^n}{p}$. 
% We now verify feasibility. Since \(R_i\ge \tfrac12\) is equivalent to
% \(\sum_{k=1}^{i-1} z_k \le 1/p\), it suffices to show that
% \(\sum_{k=1}^{i-1} z_k \le (1-(1-p)^{i-1})/p \le 1/p\).
% Indeed, \(\sum_{k=1}^{i-1} z_k\) is the expected number of selections made by \(\OPT\) up to time \(i-1\); even an omniscient policy that accepts whenever possible is interrupted by independent disruptions with probability \(p\) after each acceptance, so its expected number of acceptances by time \(i-1\) is at most the truncated geometric sum
% \(\sum_{k=1}^{i-1}(1-p)^{k-1}=(1-(1-p)^{i-1})/p\).
% Hence the denominators in \eqref{eq:closed-forms} are at least \(1\), ensuring \(\epsilon_i\in[0,1]\) and \(R_i\ge\tfrac12\) for all \(i\).
To ensure $R_i\geq 1/2$ so that $\varepsilon_i$ is a valid probability, it suffices to show $\sum_{k=1}^{i-1} z_k \leq 1/p$. Bounding the expected number of $\OPT$'s selections before time $i$ gives

%Since item $i$ is selected implies that no disruption occurred until
%Bounding the expected number of $\OPT$'s selections before $i$ gives
\[
\sum_{k=1}^{i-1} z_k \leq \sum_{k=1}^{i-1} (1-p)^{k-1}
= \frac{1-(1-p)^{i-1}}{p} \leq \frac{1}{p}.
\]
Finally, combining \eqref{eq:OPT-upper} and \eqref{eq:ALG-value} with $R_i(1-\varepsilon_i)=\tfrac12$ yields
\[
v(\ALG)=\frac{1}{2}\sum_{i=1}^n B \cdot \E[X_i\mid X_i\geq \tau_i]\cdot z_i
\geq \frac{1}{2}v(\OPT).
\]
Since the optimal online policy attains a value of at least $v(\ALG)$, the result follows.
\end{proof}

 \paragraph{Upper Bound.} 
We now present an instance that provides a tight upper bound for the competitive ratio under the non-i.i.d.\ variant of \ppp. Consider $n$ items arriving in order. The last item follows the distribution $X_n=1/\varepsilon$ with probability $\varepsilon$ and $X_n=0$ otherwise, whereas all preceding items share a deterministic value: $X_1=\cdots=X_{n-1}=(p-\varepsilon)$, where $\varepsilon>0$ is sufficiently small. At the penultimate encounter (at time $n-1$), the optimal online algorithm either skips item $n-1$ and advances to accept the last item (with expected value gain of $B$), or accepts item $n-1$ and possibly advances to accept the last item (with expected value gain of $B\big[(p-\varepsilon)+(1-p)\big]=B(1-\varepsilon)$). Therefore, for any $\varepsilon>0$, it is optimal to skip in the penultimate encounter. By induction, it follows that the optimal online algorithm skips the first $n-1$ items and accepts the last item. Thus, $v(\ALG) = B$. On the other hand, the offline benchmark goes backwards in order, and retrieves: 
\begin{align*}
v(\OPT)
  &= \varepsilon \!\left(
      B\,\frac{1}{\varepsilon}
      + B(p-\varepsilon)\!\sum_{m=1}^{n-1}(1-p)^m
    \right)
     + (1-\varepsilon)B(p-\varepsilon)
       \!\sum_{m=0}^{n-2}(1-p)^m \\[4pt]
  &= B \!\left[
      1
      + \frac{(p-\varepsilon)(1-\varepsilon p)}{p}
        \bigl(1-(1-p)^{\,n-1}\bigr)
    \right].
\end{align*}

\noindent Letting $n\to\infty$ and then $\varepsilon\to 0$ yields the asymptotic competitive ratio of
\[
\lim_{\substack{n\to\infty\\ \varepsilon\to 0}}\frac{v(\ALG)}{v(\OPT)}
=\frac{1}{2}.
\]

    }

\end{document}